\documentclass[a4paper,12pt, ruled, vlined, onesided]{article}

\usepackage{subfigure}
\usepackage{amsthm}
\usepackage{amsmath, amssymb, amsfonts}

\usepackage{dsfont}
\usepackage{comment}
\usepackage{graphicx}
\usepackage{amssymb,epsfig}
\usepackage{authblk}
\usepackage[all]{xy}
\usepackage{appendix}
\usepackage[authoryear]{natbib}
\usepackage{mathrsfs}
\usepackage{ragged2e} 
\usepackage{enumitem} 
\usepackage{microtype}
\usepackage{graphicx}
\usepackage{subfigure}
\usepackage{graphics}
\usepackage[backref, colorlinks=true]{hyperref}
\hypersetup{
	linktocpage={true},
	linkcolor = {red},
	citecolor = {cyan} 
}
\usepackage[ruled, vlined, linesnumbered]{algorithm2e}
\usepackage{multirow}
\usepackage[font=footnotesize, labelfont=bf]{caption}
\usepackage{booktabs} 

\usepackage{pstricks}
\usepackage{pst-node, pst-text, pst-3d}

\newpsstyle{Cblue}{fillstyle=solid,fillcolor=blue!30}
\newpsstyle{Cred}{fillstyle=solid,fillcolor=red!30,shadow=true}
\newpsstyle{Cdet}{}
\newpsstyle{Cmisdet}{fillstyle=solid,shadow=true}

\pdfpageheight\paperheight
\pdfpagewidth\paperwidth

\newtheorem{theorem}{Theorem}
\newtheorem{proposition}[]{Proposition}%
\newtheorem{definition}{Definition}%
\newtheorem{example}{Example}%

\graphicspath{{./}{../figures/}}

\begin{document}

\title{Statistic Selection and MCMC for Differentially Private Bayesian Estimation}

\author[]{Barış Alparslan}
\author[]{Sinan Yıldırım}

\affil[]{\normalsize Faculty of Engineering and Natural Sciences, Sabancı University, İstanbul, Turkey}
\affil[]{\small{\texttt{baris.alparslan@sabanciuniv.edu, sinanyildirim@sabanciuniv.edu}}}


\maketitle

\begin{abstract}
This paper concerns differentially private Bayesian estimation of the parameters of a population distribution, when a statistic of a sample from that population is shared in noise to provide differential privacy. 

This work mainly addresses two problems: (1) What statistic of the sample should be shared privately? For the first question, i.e., the one about statistic selection, we promote using the Fisher information. We find out that, the statistic that is most informative in a non-privacy setting may not be the optimal choice under the privacy restrictions. We provide several examples to support that point. We consider several types of data sharing settings and propose several Monte Carlo-based numerical estimation methods for calculating the Fisher information for those settings. The second question concerns inference: (2) Based on the shared statistics, how could we perform effective Bayesian inference?  We propose several Markov chain Monte Carlo (MCMC) algorithms for sampling from the posterior distribution of the parameter given the noisy statistic. The proposed MCMC algorithms can be preferred over one another depending on the problem. For example, when the shared statistics is additive and added Gaussian noise, a simple Metropolis-Hasting algorithm that utilizes the central limit theorem is a decent choice. We propose more advanced MCMC algorithms for several other cases of practical relevance. 

Our numerical examples involve comparing several candidate statistics to be shared privately. For each statistic, we perform Bayesian estimation based on the posterior distribution conditional on the privatized version of that statistic. We demonstrate that, the relative performance of a statistic, in terms of the mean squared error of the Bayesian estimator based on the corresponding privatized statistic, is adequately predicted by the Fisher information of the privatized statistic.

{\footnotesize \textbf{Keywords:} Differential privacy; Markov chain Monte Carlo; Fisher Information; Statistic selection; Bayesian Statistics.}
\end{abstract}

\newpage

\section{Introduction} \label{intro}
In recent years, differential privacy  has become a popular framework for achieving privacy-preserving data sharing and inferential analysis of sensitive data sets \citep{Dwork_2006, Dwork_and_Roth_2013}.  In this paper, we are interested in differentially private Bayesian estimation for the parameters of a population distribution, when a statistic of a sample from that population is shared in noise so as to provide differential privacy. This work concerns two problems:
\begin{enumerate}
\item What statistic of the sample should be shared?
\item Based on the shared statistics, how could we make inference?
\end{enumerate}
For the first question, i.e., the one about statistic selection, we consider using the Fisher information. For the second question, we propose Markov chain Monte Carlo (MCMC) in general, to draw samples from the posterior distribution of the parameter given the noisy statistic.

Bayesian inference using MCMC has been recently studied in the data privacy context.  One of the first works concerning using Monte Carlo for differentially private posterior sampling is \citet{wang_et_al_2015}. In \citet{wang_et_al_2015}, a class of Stochastic Gradient MCMC techniques are adapted for differential privacy.  The scheme is later improved by a few works including \citet{Li_et_al_2019}.  A general purpose and scalable differentially private MCMC algorithm was proposed in \citet{Heikkila_et_al_2019}. Both \citet{wang_et_al_2015} and \citet{Heikkila_et_al_2019} lead to non-exact MCMC algorithms, in the sense that the target distribution is sampled from only asymptotically. \citet{Yildirim_and_Ermis_2019} developed an exact MCMC algorithm based on the penalty algorithm that targets the posterior distribution and provides differential privacy at the same time. The penalty algorithm is based on the classical Metropolis-Hastings algorithm; the adoption of the differentially private penalty algorithm in \citet{Yildirim_and_Ermis_2019} for Hamiltonian Monte Carlo is recently proposed in \citet{Raisa_et_al_2021}.

The above-mentioned works propose MCMC algorithms where in every iteration some function of the sensitive data is revealed in noise so that the iterations are made differentially private. Consider, for example, two parties, an analyst and a data-holder, where the analyst wishes to perform inference based on the sensitive data held secret by the data-holder. By their nature, the algorithms mentioned above require ongoing queries to the database. Although this can be available in some cases, it may not be practical in other situations due to the requirement of continuous interaction between the two parties as long as the course of the algorithm.
Instead, in such situations it may be more feasible for the data-holder to share the data in a private manner once and for all and for the analyst to perform inference based on the shared statistic without further interaction with the data-holder. In this paper we consider the latter case, i.e., the one where instead of the involved interaction between iterations, summaries of data are shared privately prior to statistical analysis.

\citet{foulds_et_al_2016} considered adding Gaussian noise to the statistics and showed the asymptotic properties of posterior distribution when the noisy statistics are used as if the true values. The differential privacy  of the generic Metropolis-Hastings algorithm is also analyzed in \citet{foulds_et_al_2016}. \citet{foulds_et_al_2016} then proposed Gibbs sampling for problems when the likelihood belongs to an exponential family. The restriction to exponential families can be indeed limiting. In theory, with advanced MCMC methodology, one can sample from the posterior distribution of a parameter when any informative statistic of the sensitive data is shared. Moreover, the method of \citet{foulds_et_al_2016} is only asymptotically biased as it does not account for the added noise to the sufficient statistics in its model.

Unlike \citet{foulds_et_al_2016}, the works of \citet{Williams_and_McSherry_2010, Karwa_et_al_2014, Bernstein_and_Sheldon_2018, Gong_2019} correctly accommodate the shared noisy statistic of the sensitive data into a hierarchical model which has the structure
\begin{equation*}
\text{parameter} \rightarrow \text{sensitive data} \rightarrow \text{noisy statistic}. 
\end{equation*}

The posterior distribution based on the noisy statistic has very strong resemblance to the already existing approximate Bayesian computation (ABC) literature. Although the hierarchical model above has been considered earlier, e.g. in \citet{Williams_and_McSherry_2010, Karwa_et_al_2014, Bernstein_and_Sheldon_2018}, the relation between differentially private statistics and approximate Bayesian computation, in particular noisy ABC, is pronounced for the first time in \citet{Gong_2019}. While \citet{Bernstein_and_Sheldon_2018} proposed Gibbs sampler for the hierarchical model and released samples from the posterior using  two stage updating process, a rejection sampler for the ABC posterior as well as an expectation-maximization (EM) algorithm is proposed in \citet{Gong_2019}. In a following work, \citet{Park_et_al_2021} developed an differentially private ABC method in  maximum mean discrepancy is used as the distance metric between artificial and observed data and the acceptance probability is randomized. 

In ABC, the artificial and real data are usually compared via a statistic. In this paper, we consider the following scenario. A statistic of the sensitive data is shared in a privacy preserving manner. Then, one samples from the noisy ABC posterior of the parameter of interest conditional on the noisy statistic. Nevertheless, the choice of the statistic is important for inference from finite data: In the case without privacy concerns, one would like to choose the statistic that is most informative about the parameter to be estimated \citep{fearnhead_and_prangle_2012}. When ABC is done in a DP context, however, the most informative statistic in the non-private setting is not necessarily the best choice. This is because the statistic is revealed in privacy preserving noise and the noise variance depends on the sensitivity of the statistic. In order to determine the best choice for the statistic to be shared, one must compare among the informativeness of the \emph{noisy} statistics. 

The question of scope and efficiency of statistical learning with differential privacy has been studied in the literature \citep{Kasiviswanathan_et_al_2008, Dwork_and_Lei_2009, Dwork_and_Smith_2010, Smith_2011, Lei_2011}. \citet{Kasiviswanathan_et_al_2008} demonstrated that, in most problems where the relation between examples to labels are learned, a private learner can learn what a non-private learner can in the same order of the number of samples. \citet{Smith_2011} established the existence of differentially private estimators with the same asymptotic variance as their non-private counterparts, also proposing such an estimator which can be seen as an improved version of those in \citet{Dwork_and_Lei_2009, Dwork_and_Smith_2010}. Both \citet{Smith_2011} and \citet{Dwork_and_Lei_2009, Dwork_and_Smith_2010} are based on the subsample and aggregate technique of \citet{Nissim_et_al_2007}. There also exist related works where robust statistics and M-estimators \citep{Lei_2011, Smith_2011, Avella-Medina_2019} are studied in a data privacy context.

Although the works mentioned above contain methods with certain convergence guarantees, they are not directly concerned with choosing the best representation of the data, either individually or in an aggregate fashion. This motivates our first contribution of the paper, which is a method for statistic selection to be used in the private data-sharing step. We propose to use the Fisher information  contained in the noisy statistic for the parameter as the criterion to compare. The Fisher information is a relevant measure when one wishes to use a likelihood-based estimation for the unknown parameter, such as maximum likelihood estimation and Bayesian estimation. The Fisher information is a function of the parameter and it depends jointly the statistic, its sensitivity, and the targeted privacy level.

The statistic selection step of private data sharing indeed bears practical importance. Take, for example, two choices for the shared statistic of a sensitive sample $X_{1:n}$ from a normal population with mean 0 and an unknown variance: one being the sample average of squares $X_{i}^{2}$ and the other being the sample average of the absolute values $\vert X_{i}\vert$. While the order of the sample size to learn the unknown variance is the same for both choices, a non-asymptotic analysis will reveal one of them preferable over the other. Indeed, in Examples \ref{ex: Mean parameter of the normal distribution}, \ref{ex: Variance parameter of the normal distribution}, and \ref{ex: Width parameter of the uniform distribution}, we show on simple distributions that the conventional statistics may not be the best choices to share the data privately.

As a second contribution, we propose effective MCMC algorithms that target the true posterior of the noisy ABC based on the noisy statistic. Obviously, there is no ideal algorithm that performs best in all the scenarios considered here. However, the broad family of exact-approximate MCMC algorithms offer effective choices. We inspect specifically algorithms that are based on pseudo-marginal MCMC \citep{andrieu_and_roberts_2009} and the recently introduced framework called Metropolis-Hastings with averaged acceptance ratio (MHAAR) in \citep{andrieu_et_al_2020}.

The organization of the paper is as follows. In Section \ref{sec: Differential privacy}, we introduce the basic concepts of differential privacy. In Section \ref{sec: Statistic selection based on Fisher information}, we discuss the problem of parameter estimation using privatized noisy statistics of the sensitive data and propose Fisher information as a measure of informativeness of the shared statistic (in noise). We show with analytical examples that, according to Fisher information, sharing non-standard statistics for a population parameter may be more beneficial compared to standard statistics. Section \ref{sec: Bayesian inference using MCMC} is reserved for the MCMC based Bayesian inference algorithms proposed for the models induced by the privacy preserving sharing scenarios described in Section \ref{sec: Statistic selection based on Fisher information}. In Section \ref{sec: Numerical examples} we present the results of some numerical experiments. Finally, we give our concluding remarks and possible future work in Section \ref{sec: Conclusion}.

\section{Differential Privacy} \label{sec: Differential privacy}

In this section, we take differential privacy as the primary definition of data privacy; although we also mention other closely related definitions.  

Let $\mathcal{X}$ be a universal set of individual data values. We call two data sets $x_{1:n}, x_{1:n}' \in \mathcal{X}^{n}$ neighbors if $x_{1:n}'$ can be obtained by changing the value of a single entry in $x_{1:n}$. In other words, the Hamming distance between the data sets, shown as $h(x_{1:n}, x_{1:n}')$ and defined as as the number of different elements between $x_{1:n}$ and $x_{1:n}'$, is equal to $1$. We call $\mathcal{A}$ a randomized algorithm whose output upon taking the input $x_{1:n}$ is a random variable $\mathcal{A}(x_{1:n})$ taking values from some $\mathcal{Y}$. 
\begin{definition}[Differential privacy]
We say that $\mathcal{A}$ is $(\epsilon, \delta)$-differentially private if, for any pair of neighboring data sets $x_{1:n}, x_{1:n}' \in \mathcal{X}^{n}$ from an input set and any subset of output values $O \subseteq \mathcal{Y}$, it satisfies the inequality \cite{Dwork_2006}
\[
\mathbb{P} \left[ \mathcal{A}(x_{1:n})\in O \right]  \leq e^{\epsilon} \mathbb{P} \left[ \mathcal{A}(x_{1:n}')\in O \right] + \delta.
\] 
\end{definition}
According to the above inequality, a randomized algorithm is differentially private if the probability distributions for the output obtained from two neighboring databases are `{\em similar}'. The privacy parameters $(\epsilon, \delta)$ are desired to be as small as possible as far as privacy is concerned.

Assume that a privacy preserving algorithm is required to return the value of a function $\varphi:\mathcal{X} \mapsto \mathbb{R}$ evaluated at the sensitive data set $x_{1:n}$ in a private fashion. One basic way of achieving this is via the \emph{Laplace mechanism} \cite{Dwork_2008}, which relies on the \emph{(global) sensitivity} of this function.
\begin{definition}[Global sensitivity] \label{defn: global sensitivity}
The $L_{p}$ sensitivity of a function $\psi: \mathcal{X} \mapsto \mathbb{R}^{d_{\psi}}$ for $p \geq 1$ is given by
\[
\nabla_{\psi, p} = \sup_{x_{1:n}, x_{1:n}': h(x_{1:n}, x'_{1:n}) = 1} \|\psi(x_{1:n}) - \psi(x_{1:n}')\|_{p}.
\]
\end{definition}

\begin{theorem}[Laplace mechanism] \label{thm: Laplace mechanism}
Let $\mathcal{A}$ be an algorithm that returns $\psi(x_{1:n}) + V$ on an input $x_{1:n} \in \mathcal{X}^{n}$, where $V_{i} \overset{\textup{i.i.d.}}{\sim} \textup{Laplace}(\nabla_{\psi, 1}/\epsilon)$ for $i = 1, \ldots, d_{\psi}$. Then $\mathcal{A}$ is $\epsilon$-DP.
\end{theorem}
While the Laplace mechanism achieves pure differential privacy, i.e., with $\delta = 0$, another popular mechanism, called the Gaussian mechanism \citep{Dwork_and_Roth_2013} achieves differential privacy with $\delta > 0$. 
This mechanism adds Gaussian noise to $\psi(x_{1:n})$ where the variance of the noise is determined by the global sensitivity of $\psi(\cdot)$. The Gaussian mechanism is also a central tool according to other related definitions of differential privacy, such as the zero-concentrated differential privacy \citep{Bun_and_Steinke_2016} or the more recently introduced Gaussian differential privacy \citep{Dong_et_al_2022}. For an example, the following theorem presents the privacy property of the Gaussian mechanism according to Gaussian differential privacy.
\begin{theorem}[Gaussian differential privacy of the Gaussian mechanism \citep{Dong_et_al_2022}] 
Gaussian mechanism that returns $\psi(x_{1:n}) + V$, where $V_{i} \sim \mathcal{N}(0, \nabla_{\psi, 2}^{2}/\epsilon^{2})$, for $i = 1, \ldots, d_{\psi}$, satisfies $\epsilon$-Gaussian differential privacy.
\end{theorem}
A similar result regarding the Gaussian mechanism also exists for the zero-concentrated differential privacy \citep{Bun_and_Steinke_2016}. Moreover, the mentioned privacy definitions are interrelated; see \citet{Dong_et_al_2022} and \citet{Bun_and_Steinke_2016} for the detailed relations.

One property of differential privacy relevant to our work is the post-processing property, which simply holds that the privacy loss is not increased by transforming the output through an algorithm independent of the private data given the output.
\begin{theorem}[Post-processing]
Let $\mathcal{A}_{1}$ be an $(\epsilon, \delta)$-DP algorithm with inputs from  $\mathsf{X}$ and outputs from $\mathcal{S}_{1}$, and let $\mathcal{A}_{2}: \mathcal{S}_{1} \mapsto \mathcal{S}$ be an algorithm that does not depend on $X$. Then, the algorithm $\mathcal{A} = (A_{2} o \mathcal{A}_{1})$ is $(\epsilon, \delta)$-DP.
\end{theorem}
In this work, all the Bayesian inference algorithms in Section \ref{sec: Bayesian inference using MCMC} act as post-processing operations.

\section{Statistic selection based on Fisher information} \label{sec: Statistic selection based on Fisher information}

We consider a data privacy setting where we have some sensitive data $X_{1}, \ldots, X_{n} \overset{\text{i.i.d}}{\sim} \mathcal{P}_{\theta}$ for some distribution $\mathcal{P}_{\theta}$ on $\mathcal{X}$ with parameter $\theta \in \Theta$. We assume that each $X_{i}$ belongs to a distinct individual. We aim to infer $\theta$ based on the outputs of a differentially private operation on the sensitive data $X_{1:n}$.  One example case is when a statistic of the data $S_{n}: \mathcal{X}^{n} \mapsto \mathbb{R}^{d_{s}}$, for some $d_{s} \geq 1$,  is released in noise as
\begin{equation} \label{eq: noisy statistic}
Y = S_{n}(X_{1:n}) + V, \quad V \sim \mathcal{P}_{\epsilon, S_{n}},
\end{equation}
where $\mathcal{P}_{\epsilon, S_{n}}$ is the distribution of the privacy preserving noise $V$ whose parameter(s) is (are) adjusted according to $S_{n}$ and $\epsilon$. We will show two examples of $\mathcal{P}_{\epsilon, S_{n}}$ in Sections \ref{sec: Fisher information with additive statistic and Gaussian noise} and \ref{sec: Fisher information with additive statistic and non-gaussian noise}, which arise from the Gaussian and Laplace mechanisms, respectively. A common choice of $S_{n}$ is an additive statistic as in
\begin{equation} \label{eq: additive statistic}
S_{n}(X_{1:n}) = \frac{1}{n} \sum_{i = 1}^{n} s(X_{i}).
\end{equation}
However, in this paper we will consider non-additive statistics as well.

In \eqref{eq: noisy statistic}, (a statistic of) the collected data is released in batch manner.  An alternative to that is when each $X_{i}$ is shared privately as
\begin{equation} \label{eq: sequential release}
Y_{i} = s(X_{i}) + V_{i}, \quad V_{i} \sim \mathcal{P}_{\epsilon, s}.
\end{equation}
 We will call this setting the sequential release, as opposed to the batch release in \eqref{eq: noisy statistic}.

There are several other forms of differentially private data sharing, such as the exponential mechanism. In this paper we will confine the discussion to the scenarios in \eqref{eq: noisy statistic} and \eqref{eq: sequential release}. However, as it will be clear, the proposed ideas can also be adopted for other mechanisms.

How should we choose the statistic $S_{n}$ (or $s$)? We would like to make a choice (among several candidates) so that the resulting $Y$ is most `informative'. In this paper, we consider the Fisher information as the measure of the amount of `information' that $Y$ carries about $\theta$. The Fisher information not only arises naturally in frequentist contexts, but it is also relevant to Bayesian estimation, especially for big data, owing to the Bernstein-von Mises theorem \citep{Cam_1986} Under certain regularity conditions, the posterior distribution tends to have a normal distribution with covariance determined by the Fisher information.

The Fisher information at $\theta$ is determined by the population distribution $\mathcal{P}_{\theta}$, the function $S_{n}$ (or $s$), and the privacy level $\epsilon$. To concretize the discussion, we confine the attention to the batch sharing setting in \eqref{eq: noisy statistic}. The marginal density of $Y = y$ given $\theta$ can be written as
\begin{equation} \label{eq: marginal distribution of y in terms of x}
p_{\epsilon, S_{n}}(y \vert \theta) = \int p_{\epsilon, S_{n}}(y \vert x_{1:n})  \prod_{i = 1}^{n} p(x_{i} \vert \theta) d x_{1:n}.
\end{equation}
The Fisher information with respect to this marginal distribution can be expressed as
\begin{align}
F(\theta) &= \mathbb{E} \left[ -\frac{\partial^{2} \log p_{\epsilon, S_{n}}(Y \vert \theta)}{\partial \theta \partial \theta^{T}}\right] \label{eq: FIM general} \\
&= \mathbb{E} \left[ \gamma_{\epsilon, S_{n}}(\theta; y) \gamma_{\epsilon, S_{n}}(\theta; y)^{T}\right], \label{eq: FIM general - 2}
\end{align}
where $\gamma_{\epsilon, S_{n}}(\theta; y)$ is the well-known score vector, defined as
\[
\gamma_{\epsilon, S_{n}}(\theta; y) = \frac{\partial \log p_{\epsilon, S_{n}}(Y \vert \theta)}{\partial \theta }.
\]
Whether $F(\theta)$ above can be calculated exactly or not and how it should be calculated approximately in the latter case depend on the nature of the statistic and/or the privacy preserving mechanism. Specifically for \eqref{eq: FIM general}, it is critical whether the statistic is additive or not, and/or the privacy preserving noise is Gaussian or not. Furthermore, the approach to calculate $F(\theta)$ also depends on whether the data is shared in a batch or sequential manner.

Clearly, $F(\theta)$ is a function of $\theta$ and one cannot know the informativeness of the selected statistic for the stochastic process in question without knowing the true value $\theta$ that governs the process. This appears to be an issue in applying the proposed strategy of choosing statistics based on $F(\theta)$. However, the proposed strategy can be useful in several ways. For example, in some cases one statistic can be shown to yield a larger $F(\theta)$ than another uniformly over the domain of $\theta$ (see Example \ref{ex: Variance parameter of the normal distribution}.) In other cases $F(\theta)$ can be combined with the prior distribution of $\theta$, say $\eta(\theta)$, to come up with an overall score such as $\int F(\theta) \eta(\theta) d\theta$. Finally, when one statistic is not uniformly better in terms of $F(\theta)$ than the other statistic and no prior information is available, an initial chunk of the data can be used to obtain a posterior distribution, which is then to be used to determine the best statistic as well as to act as the prior distribution for the rest of the data. 

In the rest of this section, we propose algorithms to (approximately) compute $F(\theta)$ under several practically relevant combinations of those mentioned conditions. Table \ref{tbl: Algorithm-model matching Fisher information} shows the scenario-algorithm matching which indicates the most suitable algorithm to compute $F(\theta)$ for the scenario considered.
\begin{table}[t]
\caption{Model-method matching for calculating $F(\theta)$}
\label{tbl: Algorithm-model matching Fisher information}
\centerline{
\begin{tabular}{ c  c  c }
\hline
\textbf{Model} & \textbf{Method} & \textbf{Requirement} \\
\hline
additive statistic, normal noise & Section \ref{sec: Fisher information with additive statistic and Gaussian noise} & $\mu_{s}(\theta)$ and $\Sigma_{s}(\theta)$ are differentiable w.r.t.\ $\theta$ \\
additive statistic, non-gaussian noise & Algorithm \ref{alg: Monte Carlo estimation of FIM} & $\mu_{s}(\theta)$ and $\Sigma_{s}(\theta)$ are differentiable w.r.t.\ $\theta$ \\
non-additive statistic & Algorithm \ref{alg: Monte Carlo estimation of FIM_non additive} & $p(x \vert \theta)$ is differentiable w.r.t.\ $\theta$ \\
sequential release & Algorithm \ref{alg: Monte Carlo estimation of FIM for sequential release} & $p(x \vert \theta)$ is differentiable w.r.t.\ $\theta$ \\
\hline
\end{tabular}
}
\end{table}

\subsection{Fisher information with additive statistic and Gaussian noise} \label{sec: Fisher information with additive statistic and Gaussian noise}
In the classical DP setting, imperfect privacy, i.e, $(\epsilon, \delta)$-DP for $\delta > 0$, can be obtained via the Gaussian mechanism \citep{Dwork_and_Roth_2013}. The Gaussian mechanism is not only a popular choice in differential privacy studies, but also \emph{the} natural choice for Gaussian differential privacy \citep{Dong_et_al_2022}, a privacy definition that leads to a more interpretable Neyman-Pearson type error analysis than the classical differential privacy.

The Gaussian mechanism is a noise adding mechanism which can be described generally as
\begin{equation} \label{eq: noisy statistic - Gaussian}
Y = S_{n}(X_{1:n}) + V, \quad V \sim 
\mathcal{N}(0, \sigma_{s, n, \epsilon}^{2} I).
\end{equation}
For ease of exposition, one can take $\sigma_{s, n, \epsilon}^{2} = \Delta_{s, 2}^{2}/(n^{2}\epsilon^{2})$, where, recall that, $\Delta_{s, 2}$ is the $L_{2}$ sensitivity of $s(\cdot)$. Here, the parameter $\epsilon$ has a slightly different meaning than in the definition of classical differential privacy. Specifically, the above choice for the noise distribution provides $\epsilon$-Gaussian DP and \emph{not} $\epsilon$-DP. We could make the variance also depend on a $\delta > 0$ parameter to provide $(\epsilon, \delta)$-DP, but this would distract the main messages of the discussion.

Suppose that $S_{n}$ is additive as in \eqref{eq: additive statistic}. Then one can employ a normal approximation for the distribution of $S_{n}(X_{1:n})$, along the lines of \citet{Bernstein_and_Sheldon_2018}. Let 
\[
\mu_{s}(\theta) = \mathbb{E}_{\theta}[s(X)], \quad \Sigma_{s}(\theta) = \text{Var}_{\theta}[s(X)]
\]
be the mean and covariance of $s(X)$. For large $n$, the additive statistic approximately has a normal distribution
\begin{equation} \label{eq: normal approximation for additive statistic}
S_{n}(X_{1:n}) \sim \mathcal{N}(\mu_{s}(\theta), \Sigma_{s}(\theta)/n),
\end{equation}
Combining \eqref{eq: normal approximation for additive statistic} with \eqref{eq: noisy statistic - Gaussian}, the marginal distribution of $Y$ is approximated as
\begin{equation} \label{eq: FIM for normal approximation}
Y \sim \mathcal{N}\left(\mu_{s}(\theta), \Sigma_{s}(\theta)/n + \sigma_{s, n, \epsilon}^{2} I \right).
\end{equation}
Finally, considering the transformation 
\[
\theta \mapsto \left[ \mu_{s}(\theta), \Sigma_{s}(\theta)/n + \sigma_{s, n, \epsilon}^{2} I \right],
\]
the $(i, j)$'th element of $F(\theta)$ for the distribution in \eqref{eq: FIM for normal approximation} is given by
\begin{equation*}
[F(\theta)]_{i, j} = \frac{\partial \mu_{s}(\theta)^{T}}{\partial \theta_{i}} H_{s, \epsilon, n}(\theta)^{-1}  \frac{\partial \mu_{s}(\theta)}{\partial \theta_{j}} + \frac{\text{tr}(G)}{2} 
\end{equation*}
where $H_{s, \epsilon, n}(\theta) := \frac{\Sigma_{s}(\theta)}{n} + \sigma_{s, n, \epsilon}^{2} I$ is the covariance of $Y$ and 
\[
G = \frac{1}{n^{2}}\left( H_{s, \epsilon, n}(\theta)^{-1} \frac{\partial \Sigma_{s}(\theta)}{\partial \theta_{i}}  H_{s, \epsilon, n}(\theta)^{-1} \frac{\partial \Sigma_{s}(\theta)}{\partial \theta_{j}} \right).
\]

Examples \ref{ex: Mean parameter of the normal distribution}, \ref{ex: Variance parameter of the normal distribution}, and \ref{ex: Width parameter of the uniform distribution} demonstrate how the proposed scheme can be used in simple but common inference problems.
 
\begin{example}[Mean of the normal distribution] \label{ex: Mean parameter of the normal distribution}
Assume that $\mathcal{X} = (0, A)$ and the considered population distribution for $X$ is $\mathcal{P}_{\theta} = \mathcal{N}(\theta, 1)$. Here $A$ is a number which arises due to the nature of the data generation process, large enough to have negligible effect on the distribution of $X$ (The same will be assumed in the other examples in the paper.) For statistic selection, one may want to use $s(x) = x^{a}$, where $a$ is an odd integer. Let us compare $a = 1$ and $a = 3$. We have 
\[
\mu_{s}(\theta) =  \begin{cases} \theta & \text{for } a = 1 \\ \theta^{3} + 3\theta & \text{for } a = 3 \end{cases}, \quad \Sigma_{s}(\theta)  = \begin{cases} 1 &  \text{for } a = 1 \\ 9 \theta^{4} + 36 \theta^{2} + 15   & \text{for } a = 3 \end{cases},
\]
which are differentiable w.r.t.\ $\theta$ with derivatives straightforward to calculate. With the Gaussian mechanism, the variance of $Y$ becomes $H_{s, \epsilon, n} = \Sigma_{s}(\theta)/n + A^{2a}/(n^{2} \epsilon^{2})$.

Figure \ref{fig: Normal_mean_A_10_n_100_eps_100} compares $F(\theta)$ for $a = 1$ and $a = 3$, separately for $\epsilon = 1$ and $\epsilon = \infty$ corresponding to the non-private case, with $n = 100$ and $A = 10$. As it can be observed, while $s(x) = x$ is always better in the non-private case, in the private case (when $\epsilon = 1$) the choice $s(x) = x^{3}$ seems better for larger values of $\theta$.
\begin{figure}
\centerline{
\includegraphics[scale = 1]{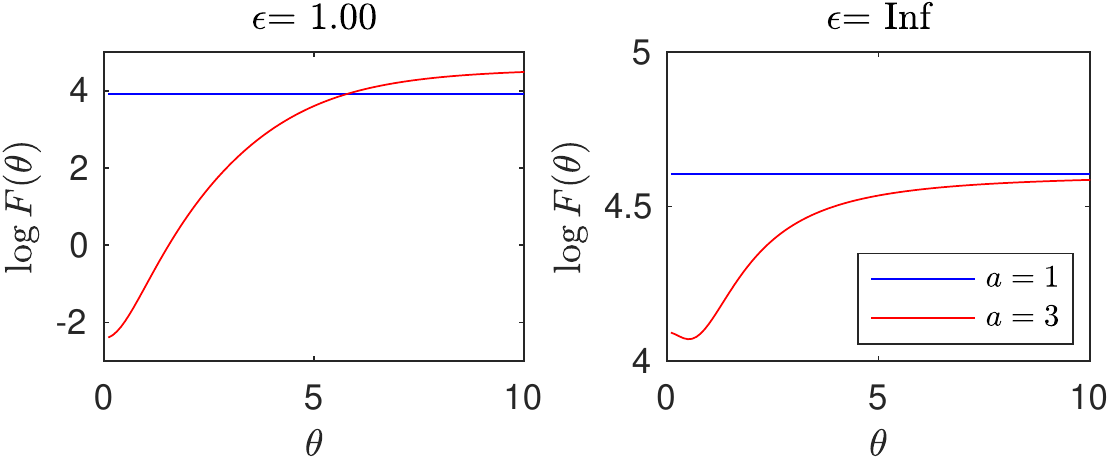}
}
\caption{$F(\theta)$ for the mean parameter of $\mathcal{N}(\theta, 1)$ when $s(x) = \vert x\vert^{a}$. Left: $\epsilon = 1$, Right: $\epsilon = \infty$ (non-private case).}
\label{fig: Normal_mean_A_10_n_100_eps_100}
\end{figure}
\end{example}

\begin{example}[Variance of the normal distribution] \label{ex: Variance parameter of the normal distribution}
Assume that $\mathcal{X} = (-A, A)$ and the assumed population distribution for $X$ is $\mathcal{P}_{\theta} = \mathcal{N}(0, \theta)$. Consider $s(x) = \vert x\vert^{a}$. So 
\begin{align*}
\mu_{s}(\theta) &= (2\theta)^{a/2} \frac{1}{\sqrt{\pi}}\Gamma\left(\frac{a+1}{2}\right), \\
\Sigma_{s}(\theta) &= (2\theta)^{a} \left[\frac{1}{\sqrt{\pi}}\Gamma\left(\frac{2a+1}{2}\right) -\frac{1}{\pi} \Gamma^{2}\left(\frac{a+1}{2}\right)\right],
\end{align*}
which are differentiable w.r.t.\ $\theta$. With the Gaussian mechanism, we have $H_{s, \epsilon, n} = \Sigma_{s}(\theta)/n + A^{2a}/(n^{2} \epsilon^{2})$.

Figure \ref{fig: Normal_variance_A_10_n_100} compares $F(\theta)$ for various values of $a$, separately for $\epsilon = 1$ and $\epsilon = \infty$ corresponding to the non-private case, with $n = 100$ and $A = 100$. As it can be observed, while $s(x) = x^{2}$ is always better in the non-private case, in the private case (when $\epsilon = 1$), the best choice is $s(x) = \vert x\vert$.

\begin{figure}[ht]
\centerline{
\includegraphics[scale = 1]{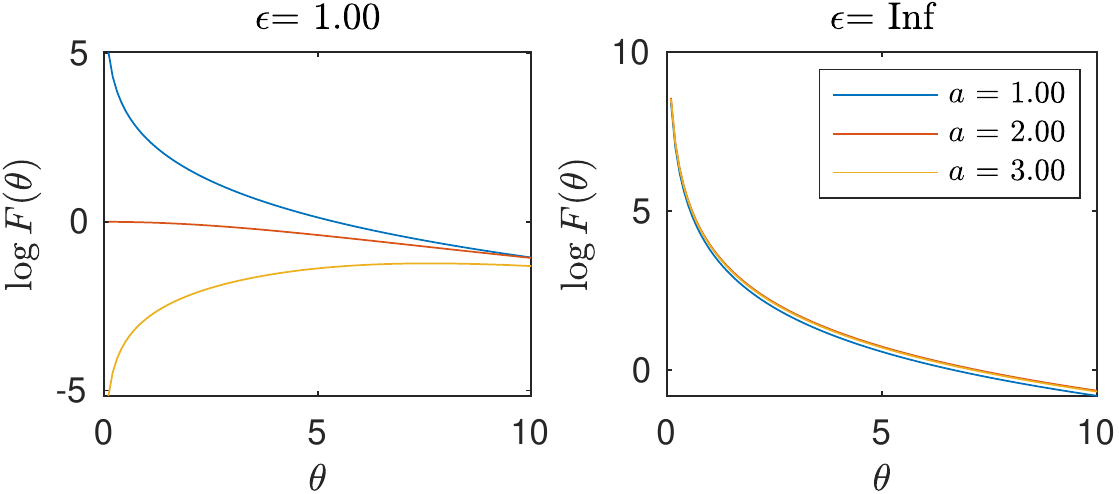}
}
\caption{$F(\theta)$ for the variance parameter of $\mathcal{N}(0, \theta)$ when $s(x) = \vert x\vert^{a}$. Left: $\epsilon = 1$, Right: $\epsilon = \infty$ (non-private case).}
\label{fig: Normal_variance_A_10_n_100}
\end{figure}
\end{example}

\begin{example}[Width of the uniform distribution] \label{ex: Width parameter of the uniform distribution}
Let $\mathcal{P}_{\theta} = \textup{Unif}(-\theta, \theta)$ so that $2 \theta$ is the width parameter of the uniform distribution. Assume that $s(x) = \vert x\vert^{a}$ for some $a > 0$. We have 
\[
\mu_{s}(\theta) = 
\frac{\theta^{a}}{a+1}, \quad \Sigma_{s}(\theta) = \frac{\theta^{2a} a^{2}}{(a+1)^{2} (2 a + 1)},
\]
which are differentiable w.r.t.\ $\theta$. Assume that $\mathcal{X} = (-A, A)$. Then the sensitivity of $s$ is $A^{a}$, hence $\Delta_{s, n} = A^{a}/n$. This yields that $H_{s, \epsilon, n}(\theta) = \Sigma_{s}(\theta)/n + A^{2a}/(n^{2}\epsilon^{2})$. 

Figure \ref{fig: Uniform_width_A_10_n_100} compares $F(\theta)$ for several values of $a$, separately for $\epsilon = 1$ and $\epsilon = \infty$ corresponding to the non-private case, with $n = 100$ and $A = 100$.
\begin{figure}[ht]
\centerline{
\includegraphics[scale = 1]{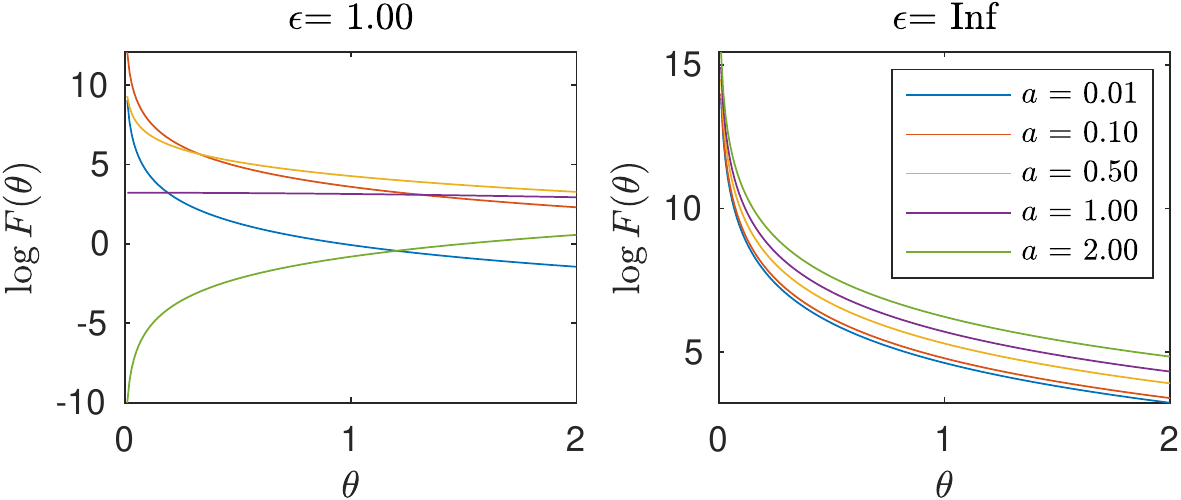}
}
\caption{$F(\theta)$ for the width parameter of $\text{Unif}(-\theta, \theta)$ when $s(x) = \vert x\vert^{a}$. Left: $\epsilon = 1$, Right: $\epsilon = \infty$ (non-private case).}
\label{fig: Uniform_width_A_10_n_100}
\end{figure}

\end{example}
Example \ref{ex: Width parameter of the uniform distribution} reveals that while $F(\theta)$ does not exist for the width parameter of the uniform distribution, it does exist for the marginal distribution of $Y$ as long as $\mu_{s}(\theta)$ $\Sigma_{s}(\theta)$ are differentiable with respec to $\theta$, thanks to the normal approximation of the distribution of the statistic. This is a promising fact for the breadth of models where the proposed methodology for statistic selection applies.

\subsection{Fisher information with additive statistics and non-gaussian noise} \label{sec: Fisher information with additive statistic and non-gaussian noise}
In the previous section, the Gaussian mechanism enabled us to perform an analytical comparisons between statistics. For other privacy preserving mechanisms, comparisons based on $F(\theta)$ can still be made in the same spirit, however by using Monte Carlo estimates of $F(\theta)$, as we will see next.

\begin{algorithm}[t]
\caption{Monte Carlo estimation of $F(\theta)$ for \eqref{eq: noisy statistic} - normal approximation for $f_{S_{n}}(u \vert \theta)$.}
\label{alg: Monte Carlo estimation of FIM}
\KwIn{$\theta$: parameter; $n$: data size; $N$, $M$: Monte Carlo parameters}
\KwOut{$\widehat{F(\theta)}$: Estimate of $F(\theta)$}
\For{$i = 1, \ldots, N$}{
Sample $y^{(i)} \sim p_{\epsilon, S_{n}}(y\vert\theta)$\\
\For{$j = 1, \ldots, M$}{
Sample $u^{(j)} \sim q_{\theta}(\cdot)$, calculate 
\[
w_{j} =\frac{ f_{S_{n}}(u^{(j)} \vert \theta) g_{\epsilon, S_{n}}(y^{(i)} \vert u^{(j)})}{ q_{\theta}(u^{(j)})}
\]
 using \eqref{eq: normal approximation of f}.
}
Using \eqref{eq: normal approximation of f}, calculate 
\[
\widehat{\gamma}_{\epsilon, S_{n}}(\theta; y^{(i)}) = \sum_{j = 1}^{N} \frac{\partial \log f_{S_{n}}(u^{(j)} \vert \theta) }{\partial \theta} \frac{w_{j}}{\sum_{j' = 1}^{N} w_{j'}}.
\]
}
\Return $\widehat{F(\theta)} =\frac{1}{N} \sum_{i = 1}^{N} \widehat{\gamma}_{\epsilon, S_{n}}(\theta; y^{(i)}) \widehat{\gamma}_{\epsilon, S_{n}}(\theta; y^{(i)})^{T}.$
\end{algorithm}

A typical example to a non-gaussian mechanism is the Laplace mechanism. In the Laplace mechanism to provide $\epsilon$-DP, privacy preserving noise in \eqref{eq: noisy statistic} is distributed according to
\[
V \sim \text{Laplace}(\Delta_{s, 1}/(n \epsilon)).
\]
As long as $S_{n}$ is an additive statistic, we can employ the normal approximation in \eqref{eq: normal approximation for additive statistic} for its distribution. Even so, the approximation of $F(\theta)$ given in Section \ref{sec: Fisher information with additive statistic and Gaussian noise} may not be accurate when non-gaussian noise is used to preserve privacy. Furthermore, the integral in \eqref{eq: marginal distribution of y in terms of x} will be typically intractable, as well as the derivative of its logarithm.  As a result, it may also be difficult to calculate $F(\theta)$ exactly. Fortunately, a consistent Monte Carlo estimator of $F(\theta)$ is available. We present its details in the following. 

Define the variable $U = S_{n}(X_{1:n})$, and let $f_{S_{n}}(u \vert \theta)$ be the probability density of $U$ given $\theta$ evaluated at $u$. In most of the models considered in this paper, the conditional distribution $p_{\epsilon, S_{n}}(y \vert x_{1:n})$ depends only on $u = S_{n}(x_{1:n})$. (See the discussion about smooth sensitivity in Section \ref{sec: Fisher information for based on the true marginal distribution} for an exception to this.). If that is the case, we can define $g_{\epsilon, S_{n}}(y \vert u)$ to be the density of the conditional distribution of $Y$ given $U$ calculated at $Y = y, U = u$. Then the marginal distribution can also be written as
\begin{equation} \label{eq: marginal distribution of y}
p_{\epsilon, S_{n}}(y \vert \theta) = \int g_{\epsilon, S_{n}}(y\vert u) f_{S_{n}}(u \vert \theta) d u.
\end{equation}

Based on \eqref{eq: marginal distribution of y}, the Fisher's identity for the score vector can be written as
\[
\gamma_{\epsilon, S_{n}}(\theta; y) = \int \frac{\partial \log f_{S_{n}}(u \vert \theta)}{\partial \theta} p( u \vert y, \theta) du. 
\]
where the integral is taken with respect to the posterior distribution 
\[
p( u \vert y, \theta) \propto f_{S_{n}}(u \vert \theta) g_{\epsilon, S_{n}}(y\vert u).
\]
The Monte Carlo estimation of $F(\theta)$ is based on estimating the above integral via importance sampling, exact sampling (e.g.\ via rejection sampling) or approximate sampling (via MCMC)  from $p( u \vert y, \theta)$. Once we have a method for obtaining a numerical approximation of the score vector at given $y$ and $\theta$, $F(\theta)$ can be estimated according to \eqref{eq: FIM general - 2}.

A Monte Carlo estimator of $F(\theta)$ in the presence of additive statistic and non-gaussian noise is given in Algorithm \ref{alg: Monte Carlo estimation of FIM}. The estimator is based on the estimation of the score vector using self-normalised importance sampling with proposal distribution $q_{\theta}(u)$. Further, the normal approximation in \eqref{eq: normal approximation for additive statistic} is employed, enabling
\begin{equation} \label{eq: normal approximation of f}
f_{S_{n}}(u \vert \theta) = \mathcal{N}(u; \mu_{s}(\theta), \Sigma_{s}(\theta)/n)
\end{equation}
in the calculations. Sampling from $p_{\epsilon, S_{n}}(y \vert \theta)$ can be performed straightforwardly since the model for $Y$ is generative. Also, the importance sampling stage (the inner loop) can be replaced by a MCMC routine to collect $M$ samples with equal weights for $u$ from the conditional distribution $p(u \vert y, \theta)$ and estimate the score by  $\frac{1}{M}\sum_{j = 1}^{M} \frac{\partial \log f_{S_{n}}(u \vert \theta)}{\partial \theta}$.

\subsection{Fisher information based on the true marginal distribution} \label{sec: Fisher information for based on the true marginal distribution}
Note that Algorithm \ref{alg: Monte Carlo estimation of FIM} exploits the normal approximation in \eqref{eq: normal approximation for additive statistic} for the statistic $S(X_{1:n})$. In some cases, this approximation may be unavailable or unjustifiable. This may be because $S_{n}(X_{1:n})$ is a \emph{non-additive} statistic, or the moments $\mu(\theta)$ and $\Sigma(\theta)$ are intractable.

In such cases, one can still devise a Monte Carlo method to estimate $F(\theta)$ based on the true marginal distribution of the observed (noisy) statistic in \eqref{eq: noisy statistic}. Such a method is given in Algorithm \ref{alg: Monte Carlo estimation of FIM_non additive}.  Algorithm \ref{alg: Monte Carlo estimation of FIM_non additive} exploits \eqref{eq: marginal distribution of y in terms of x}, which expresses the marginal distribution in terms of $X_{1:n}$. The reason we resorted to $X_{1:n}$ instead of $U = S_{n}(X_{1:n})$ is that in this part we are concerned with a setting where the probability distribution of $U$ is hard to find or approximate. Accordingly, the algorithm samples $X_{1:n}$ from their population distribution and uses importance sampling to calculate the score vector as an expectation of the derivative of the log-joint density of $(X_{1:n}, Y)$ with respect to the posterior distribution of $X_{1:n}$ given $Y$. As a result a requirement is that the population distribution is differentiable with respect to $\theta$. 

\begin{algorithm}[t]
\caption{Monte Carlo estimation of Fisher information for \eqref{eq: noisy statistic} - exact marginal distribution}
\label{alg: Monte Carlo estimation of FIM_non additive}
\KwIn{$\theta$: parameter; $n$: data size; $N$, $M$: Monte Carlo parameters}
\KwOut{$\widehat{F(\theta)}$: Estimate of $F(\theta)$}
\For{$i = 1, \ldots, N$}{
Sample $y^{(i)} \sim p_{\epsilon, S_{n}}(y \vert \theta)$\\
\For{$j = 1, \ldots, M$}{
\For{$t = 1, \ldots, n$}{
Sample $x_{t}^{(j)} \sim p(x \vert \theta)$.
}

Set $w_{j} = p_{\epsilon, S_{n}}(y^{(i)} \vert x_{1:n}^{(j)})$.
}
Calculate 
\[
\widehat{\gamma}_{\epsilon, S_{n}}(\theta; y^{(i)}) = \sum_{j = 1}^{N}  \left( \sum_{t = 1}^{n} \frac{\partial\log p(x_{t}^{(j)} \vert \theta)}{\partial \theta} \right) \frac{w_{j}}{\sum_{j' = 1}^{N} w_{j'}}.
\]
}
\Return $\widehat{F(\theta)} = \frac{1}{N} \sum_{i = 1}^{N} \widehat{\gamma}_{\epsilon, S_{n}}(\theta; y^{(i)}) \widehat{\gamma}_{\epsilon, S_{n}}(\theta; y^{(i)})^{T}$.
\end{algorithm}

At this point it is worth pointing to smooth sensitivity \citep{Nissim_et_al_2007}, a method that has proven quite useful in reducing privacy preserving noise considerably, especially for non-additive statistics, which are under consideration in this section. As before, one could use the global sensitivity of $S_{n}$ as in Definition \ref{defn: global sensitivity} to determine the amount of noise to generate $Y$. However, for some non-additive statistics, such as $\max$ and $\text{median}$, adding noise based on the global sensitivity can be quite ineffective. This is because the global sensitivity of the those functions is as large as the range of $S_{n}$. For example, if $\mathcal{X}  \mapsto [0, A]$, the global sensitivites of $\max$ and $\text{median}$ are both $A$. Instead, one can generate the noisy statistic $Y$ by adjusting the amount of noise using the smooth sensitivity defined in \citet{Nissim_et_al_2007}. 

\begin{definition}[Smooth sensitivity]
For a function $\psi: \mathcal{X}^{n} \mapsto \mathbb{R}^{d_{\psi}}$ and $\beta > 0$, define the $\beta$-smooth sensitivity as
\[
\Delta^{smooth}_{\psi, \beta}(x_{1:n}) = \max_{x'_{1:n} \in \mathcal{X}^{n}} L_{\psi}(x'_{1:n}) e^{-\beta h(x_{1:n}, x'_{1:n})},
\]
where $L_{\psi}(x_{1:n})$ is called the local sensitivity at $x_{1:n}$ and defined as
\[
L_{\psi}(x_{1:n}) = \max_{x'_{1:n}: h(x_{1:n}, x_{1:n}') = 1} \|\psi(x_{1:n})-\psi(x'_{1:n}) \|_{1}.
\]
\end{definition}

Differential privacy can be provided based on local sensitivity using appropriate noise-adding mechanisms. For example, to satisfy $(\epsilon, \delta)$-DP for $\epsilon, \delta \in (0, 1)$, one can generate $Y = S_{n}(X_{1:n}) + V$ with
\[
V \sim \text{Laplace}\left(\Delta^{smooth}_{S_{n}, \beta}(X_{1:n})/ \alpha \right),
\]
where $\alpha = \epsilon/2$ and $\beta = \epsilon/[2 \ln (2 / \delta)]$. (Pure differential privacy, i.e., with $\delta = 0$, can also be obtained using smooth sensitivity, however with a noise distribution whose tails decay slower than exponentially, such as Cauchy distribution.)

Note that, contrary to the earlier examples, using smooth sensitivity determines the noise distribution dependent on $X_{1:n}$, rendering a quite non-standard joint distribution, in particular a non-standard posterior distribution for the parameter of interest $\theta$. This highlights the importance of general-purpose inference methods in the privacy context such as MCMC. 

Finally, a remark on the notation. When smooth sensitivity is used, 
the density of the conditional distribution $y$ given $X_{1:n} = x_{1:n}$ depends on not only $S_{n}(x_{1:n})$ but also $x_{1:n}$ itself, since $x_{1:n}$ determines the noise variance also. To cover those cases, in Algorithm \ref{alg: Monte Carlo estimation of FIM_non additive} we resort the more general representation $p_{\epsilon, S_{n}}(y^{(i)} \vert x_{1:n}^{(j)})$ to denote the conditional distribution of $y$ given $x_{1:n}$.

\subsection{Fisher information with sequential release} \label{sec: Fisher information with sequential release}
In Sections \ref{sec: Fisher information with additive statistic and Gaussian noise}-\ref{sec: Fisher information for based on the true marginal distribution} we looked at scenarios where a single statistic of the sensitive data $X_{1:n}$ is shared. Alternatively, private data can be sequentially released as $Y_{1}, \ldots, Y_{n}$, where each $Y_{i}$ is a noisy version of $s(X_{i})$ as in \eqref{eq: sequential release}. This corresponds to a scenario where the analyst collects data from the individuals separately in a privacy preserving way. The former and the latter models are also referred to as the centralized model and the local model \citep{Kasiviswanathan_et_al_2008}, respectively. The local model comes with the expense of adding much more noise to each $Y_{i}$ than the statistic $S(X_{1:n})$. Specifically, to provide $\epsilon$-DP with the Laplace mechanism, we must have 
\[
Y_{i} = s(X_{i}) + V_{i}, \quad V_{i} \sim \text{Laplace}(\Delta_{s, 1}/\epsilon)
\]
which no longer has the $1/n$ factor in its noise parameter.

\begin{algorithm}[t]
\caption{Monte Carlo estimation of $F(\theta)$ for \eqref{eq: sequential release}}
\label{alg: Monte Carlo estimation of FIM for sequential release}
\KwIn{$\theta$: parameter; $n$: data size; $N$, $M$: Monte Carlo parameters}
\KwOut{$\widehat{F(\theta)}$: Estimate of $F(\theta)$}
\For{$i = 1, \ldots, N$}{
Sample $y^{(i)} \sim p_{\epsilon, s}(y \vert \theta)$\\
\For{$j = 1, \ldots, M$}{
Sample $x^{(j)} \sim q_{\theta}(x)$ and calculate 
\[
w_{j} = p(x^{(j)} \vert \theta) g_{\epsilon, s}(y^{(i)} \vert s(x^{(j)})) / q_{\theta}(x^{(j)}).
\]
}
Calculate 
\[
\widehat{\gamma}_{\epsilon, s}(\theta; y^{(i)}) = \sum_{j = 1}^{N} \frac{\partial \log p(x^{(j)} \vert \theta)}{\partial \theta} \frac{w_{j}}{\sum_{j' = 1}^{N} w_{j'}}.
\]
}
\Return $\widehat{F(\theta)} = \frac{n}{N} \sum_{i = 1}^{N} \widehat{\gamma}_{\epsilon, S_{n}}(\theta; y^{(i)}) \widehat{\gamma}_{\epsilon, S_{n}}(\theta; y^{(i)})^{T}$.
\end{algorithm}

In the local model, we can talk about the marginal distribution of each $Y_{i}$, whose probability density can be written as 
\begin{equation}\label{eq: marginal distribution sequential release}
p_{\epsilon, s}(y \vert \theta) = \int p(x\vert\theta) g_{\epsilon, s}(y \vert s(x)) dx,
\end{equation}
where $g_{\epsilon, s}(y \vert s(x))$ is the probability density function of the conditional distribution of $Y$ given $S(X)$, which, according to \eqref{eq: sequential release}, reduces to the probability of $\mathcal{P}_{\epsilon, s}$ evaluated at $y - s(x)$.

The Fisher information corresponding to this mechanism can be numerically calculated by estimating the Fisher infrormation of a single $Y_{i}$ via Monte Carlo as in Algorithm \ref{alg: Monte Carlo estimation of FIM for sequential release}. The algorithm requires that the probability density (mass) function of $X_{i}$ is differentiable w.r.t.\ $\theta$.

\begin{example}[Binary responses] \label{ex: Binary responses}
Let $X_{i} \in \{ 0, 1\}$ with $X_{i} \overset{\text{iid}}{\sim} \text{Bern}(\theta)$ for $i = 1, \ldots, n$. In a non-private setting, a natural estimator for $\theta$ is $\bar{X}$. Instead, we consider estimating $\theta$ privately. We will compare three mechanisms.
\begin{enumerate}
\item It is well known, and can be easily verified that releasing the randomized binary responses $Y_{1}, \ldots, Y_{n}$, where 
\[
Y_{i} = \begin{cases} X & \text{with probability} \frac{e^{\epsilon}}{e^{\epsilon} + 1} \\
1 - X_{i} & \text{ else}  \end{cases}
\]
provides $\epsilon$-DP. The probability of the randomized response being $1$ is given by 
\[
\tau := \mathbb{P}(Y = 1) = \frac{\theta e^{\epsilon} + (1 - \theta)}{1 + e^{\epsilon}}.
\]
The probability density of $Y$ given $\theta$ and $\epsilon$ is
 \[
\log p(y \vert \theta) = y \ln \tau + (1- y) \ln \tau.
\]
Therefore, letting $\alpha = (e^{\epsilon}-1)/(e^{\epsilon} + 1)$, the Fisher information of $Y_{1}, \ldots, Y_{n}$ is given by 
\[
F_{1}(\theta) = n \mathbb{E}\left [-\frac{\partial^{2} \log p(Y \vert \theta)}{\partial \theta^{2}} \right] 
= \frac{n\alpha^{2}}{\tau (1 - \tau)}.
\]

\item One alternative to the above is to release $Z_{i} = X_{i} + V_{i}$, with $V_{i} \overset{\text{i.i.d}}{\sim} \mathcal{N}(0, 1/\epsilon^{2})$.  It is obvious that $Z_{1}, \ldots, Z_{n}$ is as informative as $\hat{\theta}_{2} = \bar{Z}$, which approximately has the normal distribution $\mathcal{N}(\theta, \theta( 1- \theta)/n + 1/(\epsilon^{2} n))$ hence its Fisher information is approximately 
\[
F_{2}(\theta) = \frac{n (\theta(1 - \theta) + 1/\epsilon^{2}) + (1 - 2\theta)^{2}}{[\theta(1 - \theta) + 1/\epsilon^{2}]^{2}}.
\]
\item Finally, we consider adding noise to the mean $\bar{X}$ and obtain $\hat{\theta}_{3} = \bar{X} + V$, where $V \sim \mathcal{N}(0, 1/(n^{2}\epsilon^{2}))$. Therefore, $\hat{\theta}_{3} \sim \mathcal{N}(\theta, \theta( 1- \theta)/n + 1/(\epsilon^{2} n^{2}))$. This last estimator is based on a noisy average whose Fisher information is
\[
F_{3}(\theta) = \frac{n (\theta(1 - \theta) + 1/(\epsilon^{2}n)) + (1 - 2\theta)^{2}}{[\theta(1 - \theta) + 1/(\epsilon^{2} n)]^{2}}.
\]
Notice the improvement due to adding noise to the average (output) rather than averaging noisy inputs.
\end{enumerate}
Figure \ref{fig: Bernoulli example} shows a comparison between $F_{1}(\theta)$ and $F_{2}(\theta)$ as well as between $F_{1}(\theta)$ and $F_{3}(\theta)$ for $n = 100$. It can be seen that, for small values of $\epsilon$, revealing the average of the randomizing the responses is better than revealing the average of the noisy responses created by the Gaussian mechanism. However, in the same $\epsilon$ regimes, using the noisy average is better than the average of the randomised responses.

\begin{figure}[ht]
\centerline{\includegraphics[scale = 0.8]{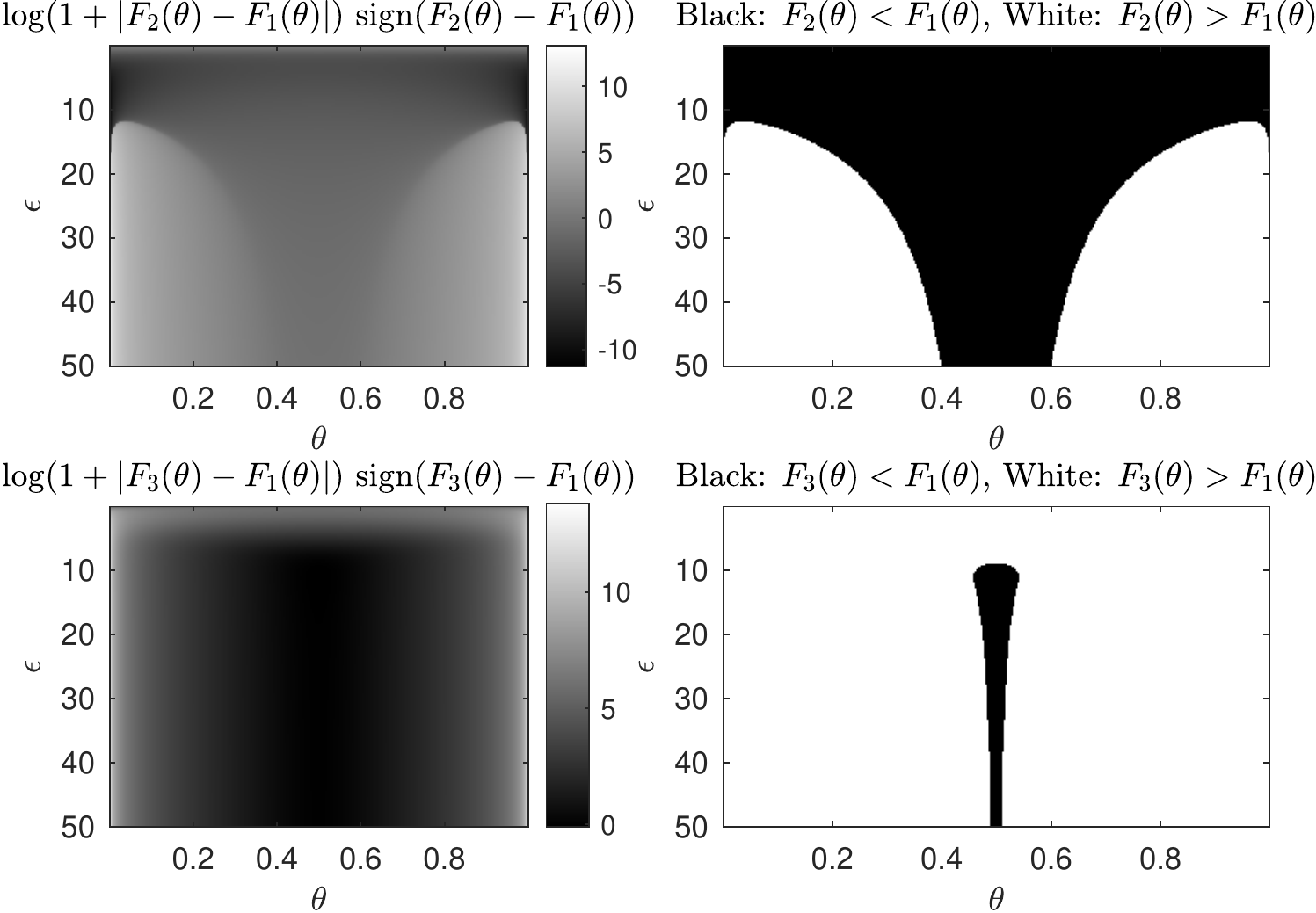} }
\caption{Comparison among $F_{1}(\theta)$, $F_{2}(\theta)$, $F_{3}(\theta)$.}
\label{fig: Bernoulli example}
\end{figure}
\end{example}

\paragraph{Graphical summary:} Figure \ref{fig: FIM graph matching} shows graphical representations of the models respected by the $F(\theta)$ calculations in this section. Note that the graphs (1-3) correspond to the same batch model in \eqref{eq: noisy statistic}, but represented with different sets of variables, while graph (4) corresponds to the model with sequential release in \eqref{eq: sequential release}. Moreover, the graphs (1-4) correspond to the MCMC algorithms in Sections \ref{sec: MH for additive statistic and Gaussian noise}-\ref{sec: Exact inference based on sequential releases}, respectively.

\begin{figure}[h]
\centerline{\includegraphics[scale = 0.9]{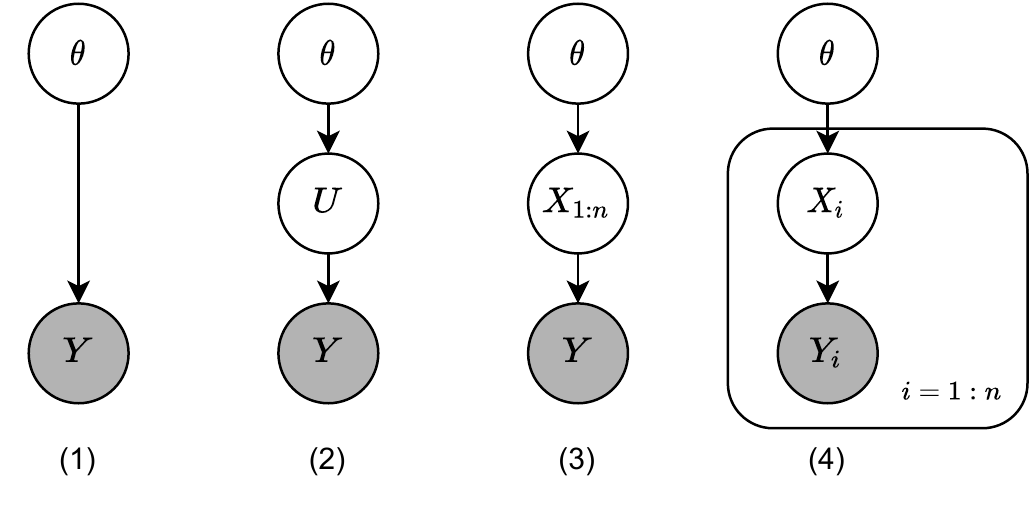}}
\caption{Graphical representations of the models respected in the $F(\theta)$ calculations and MCMC algorithms. Models (1-4) refer to Sections \ref{sec: Fisher information with additive statistic and Gaussian noise}-\ref{sec: Fisher information with sequential release}, respectively, and correspond to the MCMC algorithms in Sections \ref{sec: MH for additive statistic and Gaussian noise}-\ref{sec: Exact inference based on sequential releases}, respectively. Shaded variables are observed; the others are latent.}
\label{fig: FIM graph matching}
\end{figure}

\section{Bayesian inference using MCMC} \label{sec: Bayesian inference using MCMC}
\begin{table}[t]
\caption{Algorithm-model matching for differentially private Bayesian learning via MCMC}
\label{tbl: Algorithm-model matching}
\centerline{
\begin{tabular}{c c}
\toprule
\textbf{Model} & \textbf{Algorithm}  \\
\midrule
Additive statistic, Gaussian noise & Algorithm \ref{alg: MH for ABC-DP} \\
Additive statistic, non-gaussian noise & Algorithm \ref{alg: PMMH for ABC-DP}, \ref{alg: MHAAR for ABC-DP} \\
Non-additive statistic & Algorithm \ref{alg: MHAAR for ABC-DP with z} \\
Sequential release & Algorithm \ref{alg: MHAAR for ABC-DP sequential} \\
\hline
\end{tabular}
}
\end{table}

For the statistic selection method to be useful in practice, it should be accompanied with an inference method. Within the Bayesian framework, one could use ideas from approximate Bayesian computation. This relation is already observed in \citet{Gong_2019}, where an EM algorithm is presented for maximum likelihood estimation of $\theta$. The EM algorithm can be somewhat restrictive, for its E- and M- steps may require exact posterior expectations. An alternative to EM is to consider Bayesian estimation by means of sampling from the posterior distribution of $\theta$ given the shared statistics. Owing to the availability of Monte Carlo techniques for sampling from various forms of posterior distributions, Bayesian estimation is usually less demanding about the nature of the model in hand.

In the batch setting, where a statistic $S_{n}(X_{1:n})$ is shared in noise as in \eqref{eq: noisy statistic}, the posterior distribution is
\begin{equation} \label{eq: posterior distribution batch}
p_{\epsilon, S_{n}}(\theta \vert y) \propto \eta(\theta) p_{\epsilon, S_{n}}(y \vert \theta),
\end{equation}
where $\eta(\theta)$ is the probability density of the prior distribution of $\theta$ and the likelihood $p_{\epsilon, S_{n}}(y \vert \theta)$ is defined in \eqref{eq: marginal distribution of y}.
In the sequential setting, where a function $s$ of each $X_{i}$ is shared in noise, the posterior distribution becomes
\begin{equation} \label{eq: posterior distribution sequential}
p_{\epsilon, s}(\theta \vert y_{1:n}) \propto \eta(\theta) \prod_{t = 1}^{n} p_{\epsilon, s}(y_{t} \vert \theta).
\end{equation}
In the following, we propose MCMC algorithms for sampling from the posterior distribution for the settings investigated separately in the subsections of Section \ref{sec: Statistic selection based on Fisher information}. 

MCMC is the name for a family of methods that (approximately) sample from a given probability distribution, say $\pi(\theta)$. An MCMC algorithm is specified by an ergodic Markov chain $\{ \theta_{i}\}_{i \geq 0}$ which is designed to have $\pi(\theta)$ as its invariant distribution. In that way, the generated sequence $\{ \theta_{i}\}_{i \geq 0}$ from this Markov chain converges in distribution to $\pi(\theta)$. The methods we propose in this paper are either variants or sophisticated imitations of the Metropolis-Hastings (MH), arguably the most popular MCMC algorithm. One iteration of the MH algorithm involves (i) a proposal mechanism where, given the current value $\theta_{i} = \theta$, a candidate value $\theta'$ is proposed from the conditional distribution $q(\theta' \vert \theta)$, and (ii) an accept-reject mechanism in which the proposal is accepted and $\theta_{i+1} = \theta'$ is taken with the acceptance probability 
\[
\alpha(\theta, \theta') = \min \left\{ 1, \frac{q(\theta \vert \theta')}{q(\theta' \vert \theta)} \frac{\pi(\theta')}{\pi(\theta)} \right\};
\]
and otherwise it is rejected and the new sample is taken as the current sample, i.e., $\theta_{i+1} = \theta$. 

In the following subsections, we will propose MH-based algorithms that are suitable for each data-sharing model investigated in Sections \ref{sec: Fisher information with additive statistic and Gaussian noise}-\ref{sec: Fisher information with sequential release}.

\subsection{MH for additive statistic and Gaussian noise} \label{sec: MH for additive statistic and Gaussian noise}

Recall that when $S(X_{1:n})$ is an additive statistic as in \eqref{eq: additive statistic} and the noise is Gussian, the marginal likelihood of $Y$ given $\theta$ can be approximated as \eqref{eq: FIM for normal approximation}. We can use that approximation to obtain 
\begin{align}
\hat{p}_{\epsilon, S_{n}}(\theta \vert y) \propto \eta(\theta) \mathcal{N}(y;  \mu_{s}(\theta), H_{s, \epsilon, n}(\theta)). \label{eq: approximate posterior with gaussian noise}
\end{align}
If this distribution is intractable, MCMC can be used to sample from it. Algorithm \ref{alg: MH for ABC-DP} presents the MH algorithm for this distribution.

\begin{algorithm}
\caption{MH for \eqref{eq: approximate posterior with gaussian noise} - one iteration}
\label{alg: MH for ABC-DP}
\KwIn{Current value: $\theta$; privately shared statistic: $y$, privacy level: $\epsilon$}
\KwOut{New sample}
Propose $\theta' \sim q(\cdot \vert \theta)$\\
Accept the proposal and return $\theta'$ with probability 
\[
\min \left\{ 1,  \frac{q(\theta \vert \theta')}{q(\theta' \vert \theta)} \frac{\eta(\theta')}{\eta(\theta)} \frac{\mathcal{N}(y; \mu_{s}(\theta'), H_{s, \epsilon, n}(\theta'))}{\mathcal{N}(y; \mu_{s}(\theta), H_{s, \epsilon, n}(\theta))} \right\};
\]
otherwise reject the proposal and return $\theta$.
\end{algorithm}

\subsection{MH for additive statistic and non-gaussian noise} \label{sec: MH for additive statistic and non-gaussian noise}

Here we study inference for the setting discussed in Section \ref{sec: Fisher information with additive statistic and non-gaussian noise}, where $S_{n}(X_{1:n})$ is still additive as in \eqref{eq: additive statistic}, however a non-Gaussian mechanism (such as the Laplace mechanism) is used to preserve privacy.

Due to the additivity of $S_{n}$, we can still use the normality approximation in \eqref{eq: normal approximation for additive statistic} for $U = S_{n}(X_{1:n})$. However, due to non-gaussianity of the noise, the marginal distribution of the shared statistic $Y$ may not reliably be approximated as a normal distribution any more. 

In this model, inference can still be made via suitable MCMC algorithms. Define the joint posterior distribution
\begin{equation} \label{eq: joint dist normal approx non-gaussian noise}
\pi(\theta, u \vert y) \propto \eta(\theta) f_{S_{n}}(u \vert \theta) g_{\epsilon, S_{n}}(y \vert u)  
\end{equation}
where, recall that, $g_{\epsilon, n}(y \vert u)$ is the conditional distribution of $y$ given $u = S_{n}(x_{1:n})$. We consider sampling from this posterior distribution by using MCMC. Note that the marginal distribution of $\theta$ with respect to $\pi(\theta, u \vert y)$ can be shown to be $p_{\epsilon, S_{n}}(\theta \vert y)$ in \eqref{eq: posterior distribution batch}, which validates sampling from $\pi(\theta, u \vert y)$ as a means of sampling from $p_{\epsilon, S_{n}}(\theta \vert y)$. 

There are many possible ways to design a correct MCMC algorithm for $\pi(\theta, u \vert y)$. A standard option is to use the MH-within-Gibbs algorithm, where one iteration consists of an update of $u$ conditional on $\theta, y$ which is followed by an update of $\theta$ conditional on $u, y$. The MH-within-Gibbs algorithm may not be efficient in the presence of high dependence between the variables $\theta$ and $u$ given $y$. 

Alternative to MH-within-Gibbs, exact-approximate MCMC algorithms \citep{andrieu_and_roberts_2009, andrieu_et_al_2010, andrieu_et_al_2020} mimic the MH algorithm for the marginal posterior distribution in \eqref{eq: posterior distribution batch}. The term ``exact-approximate'' comes from the fact that the Markov chains of those algorithms still correctly converge to the exact posterior distribution (hence ``exact'') and they are approximations of the ideal (but intractable) MH algorithm for the marginal posterior distribution $p_{\epsilon, S_{n}}(\theta \vert y)$ (hence ``approximate''). Those algorithms can be useful since they circumvent the problem of dependency between $\theta$ and $u$ by relying on sample-based estimators of the marginal MH acceptance ratio. The variance of the estimator reduces with amount of computation. Moreover, the amount of computation can be mostly parallelized. 

In the following we present two exact-approximate MCMC algorithms. 
\subsubsection{Pseudo-Marginal MH}
 The pseudo-marginal MH (PMMH) of \citet{andrieu_and_roberts_2009}, adopted to the posterior distribution in \eqref{eq: posterior distribution batch} is described in Algorithm \ref{alg: PMMH for ABC-DP}. The PMMH algorithm targets the posterior distribution in \eqref{eq: joint dist normal approx non-gaussian noise}, but it mimics the MH algorithm by estimating the intractable marginal likelihood \eqref{eq: marginal distribution of y}  in \eqref{eq: posterior distribution batch} using importance sampling. Observe that the computational cost of one iteration of this algorithm is $\mathcal{O}(N)$, the sample size of the importance sampling step, which can largely be parallelised. 

\begin{algorithm}[ht]
\caption{PMMH for the posterior distribution in \eqref{eq: joint dist normal approx non-gaussian noise} - one iteration}
\label{alg: PMMH for ABC-DP}
\KwIn{Current sample: $(\theta, \hat{Z})$, number of proposals for $u$: $N$ privately shared statistic $y$}
\KwOut{New sample}
Propose $\theta' \sim q(\cdot \vert \theta)$

Sample $u^{(j)} \sim q_{\theta'}(\cdot)$ for $j = 1, \ldots, N$.

Calculate $\hat{Z}' = \frac{1}{N} \sum_{j = 1}^{N} f_{S_{n}}(u^{(j)} \vert \theta) g_{\epsilon, n}(y \vert u^{(j)}) / q_{\theta'}(u^{(j)})$ using \eqref{eq: normal approximation of f}.

Return $(\theta', \hat{Z}')$ with probability
\[
\min \left\{ 1, \frac{q(\theta \vert \theta')}{q(\theta' \vert \theta)} \frac{\eta(\theta')}{\eta(\theta)}  \frac{\hat{Z}'}{\hat{Z} }\right\};
\]
otherwise, reject and return $(\theta, \hat{Z})$.
\end{algorithm}

\subsubsection{MH with Averaged Acceptance Ratios} 
In PMMH, the estimate $\hat{Z}$ in the denominator of the acceptance ratio is carried over from the previous iteration, which can lead to stickiness in its Markov chain. The correlated pseudo-marginal algorithm of \citet{Deligiannidis_et_al_2018} partially alleviates the stickiness problem by making the numerator and denominator correlated, which is achieved by employing a common source of randomness in the estimators of the the numerator and denominator of the marginal acceptance ratio. This idea of using correlated estimators is taken to its limit by a more recent class of exact-approximate MCMC algorithms proposed in \citet{andrieu_et_al_2020} with the name ``MH with Averaged Acceptance Ratio (MHAAR)''. Unlike PMMH or its correlated version, in MHAAR both the numerator and the denominator of the acceptance ratio estimator are (almost) fully refreshed in every iteration, which is one advantage of MHAAR over PMMH.

While there are several versions of MHAAR which can be applied to the posterior distribution in \eqref{eq: joint dist normal approx non-gaussian noise}, we present a particular variant in \citet[Section 3]{andrieu_et_al_2020} in Algorithm \ref{alg: MHAAR for ABC-DP}. The requirement in Algorithm \ref{alg: MHAAR for ABC-DP} to work properly, the proposal distribution for $u$ has to satisfy $q_{\theta, \theta'}(u) = q_{\theta', \theta}(u)$ for all $\theta, \theta'$ and $u$.

\begin{algorithm}[ht]
\caption{MHAAR for the posterior distribution in \eqref{eq: joint dist normal approx non-gaussian noise} - one iteration}
\label{alg: MHAAR for ABC-DP}
\KwIn{Current value: $(\theta, u)$; number of proposals for $u$: $N$; privately shared statistic: $y$}
\KwOut{New sample}
Propose $\theta' \sim q(\cdot \vert \theta)$

\For{$j = 1, \ldots, N$}{
If $j = 1$ set $u^{(1)} = u$; otherwise sample $u^{(j)} \sim q_{\theta, \theta'}(\cdot)$.

Using \eqref{eq: normal approximation of f}, calculate 
\[
w_{j} = \frac{ f_{S_{n}}(u^{(j)} \vert \theta)  g_{\epsilon, n}(y \vert u^{(j)})}{ q_{\theta, \theta'}(u^{(j)})}, \quad 
w'_{j} = \frac{f_{S_{n}}(u^{(j)} \vert \theta')  g_{\epsilon, n}(y \vert u^{(j)})}{ q_{\theta, \theta'}(u^{(j)})}
\]
}
With probability 
\[
\min \left\{ 1,  \frac{q(\theta \vert \theta')}{q(\theta' \vert \theta)} \frac{\eta(\theta')}{\eta(\theta)}  \frac{ \sum_{j = 1}^{N} w'_{j}}{\sum_{j = 1}^{N} w_{j} }\right\},
\]
sample $k \in \{1, \ldots, N \}$ with probability proportional to $w'_{k}$ and return $(\theta', u^{(k)})$. Otherwise, reject the move, sample $k \in \{1, \ldots, N \}$ with probability proportional to $w_{k}$, and return $(\theta, u^{(k)})$.
\end{algorithm}

\subsection{Exact inference based on the true posterior} \label{sec: Exact inference based on the true posterior}

The algorithms in the previous sections can be restrictive, since $\mu_{s}(\theta)$ and $\Sigma_{s}(\theta)$ may not be analytically available, or the normality approximation may not be valid if $S_{n}$ is an extreme statistic of $X_{1:n}$, such as $S(x_{1:n}) = \max_{1 \leq t \leq n} s(x_{t})$. In such models, the true posterior of $\theta$ may be targeted via a special variable augmentation. Consider, for example, the extended posterior distribution
\[
\pi(\theta, x_{1:n} \vert y) \propto \eta(\theta) p(x_{1:n} \vert \theta) p_{\epsilon, S_{n}}(y \vert x_{1:n}).
\]
(Alternatively, one may choose to augment the space with the statistic $u = S(x_{1:n})$ and work with \eqref{eq: joint dist normal approx non-gaussian noise} if $f_{S_{n}}(u \vert \theta)$ can be calculated exactly. However we do not pursue this option to avoid diverting from the main point.)

One can go to an even lower-level representation and express the joint distribution in terms of the random variables that generate $X_{i}$'s and have distributions that do not depend on $\theta$. Letting those random variables $Z_{i} \sim \mu(\cdot)$, assume a transformation $\varphi_{\theta}(\cdot)$  such that if 
\begin{equation} \label{eq: pseudorandom variables}
Z_{i} \overset{\text{i.i.d}}{\sim} \mu(\cdot) \Rightarrow X_{i} = \varphi_{\theta}(Z_{i}) \overset{\text{i.i.d}}{\sim} \mathcal{P}_{\theta}, \quad i \geq 1
\end{equation}
Without loss of generality, $Z_{1:n}$ can be thought of a sequence of random variables from $\text{Unif}(0, 1)$, owing to the role of uniformly distributed pseudo-random variables in generation of random variables from any distribution via some suitable transformation.

Presence of such $Z_{1:n}$ induces the joint posterior distribution
\begin{equation}\label{eq: true posterior}
\pi(\theta, z_{1:n} \vert y) \propto \eta(\theta) \prod_{t = 1}^{n} \mu(z_{t}) h_{\epsilon, S_{n}}(y \vert z_{1:n}, \theta),
\end{equation}
where $h_{\epsilon, S_{n}}(y \vert z_{1:n}, \theta) = p_{\epsilon, S_{n}}(y \vert x_{1:n})$, with $x_{i} = \varphi(x_{i})$, is a re-parametrization of the conditional density in terms of $z_{1:n}$. Crucially, it can be shown that the marginal distribution for $\theta$ with respect to $\pi(\theta, z_{1:n} \vert y)$ is the target posterior distribution $p_{\epsilon, S_{n}}(\theta \vert y)$.

Choosing $Z_{1:n}$ such that its density does not depend on $\theta$ enables the MHAAR methodology of \citet{andrieu_et_al_2020}, where estimates of the acceptance ratio can efficiently be averaged to reduce variance. Algorithm \ref{alg: MHAAR for ABC-DP with z} is inspired from \citet{andrieu_et_al_2020} and can be thought of a variant of MHAAR. However, it bears methodological novelty in the sense that, if desired, only a subset of the latent variables $z_{1:n}$ can be updated per iteration instead of the whole $z_{1:n}$ (which may be computationally cheap in some cases.) Hence, Algorithm \ref{alg: MHAAR for ABC-DP with z} requires to be proven for its validity. We establish that in the following proposition. A proof is given in Appendix \ref{sec: Proof of Proposition}.
\begin{proposition} \label{prop: MHAAR variant invariance}
The Markov kernel of Algorithm \ref{alg: MHAAR for ABC-DP with z} has $\pi(\theta, z_{1:n} \vert y)$ in \eqref{eq: true posterior} as its invariant distribution.
\end{proposition}

\begin{algorithm}
\caption{MHAAR for \eqref{eq: true posterior} - one iteration}
\label{alg: MHAAR for ABC-DP with z}
\KwIn{Current sample: $(\theta, z_{1:n})$, subset size: $m < n$, number of samples for $z_{1:n}$: $N$, privately shared statistic: $y$.}
\KwOut{New sample}
Propose $\theta' \sim q(\cdot \vert \theta)$

\If{full mode}{
Set $z_{1:n}^{(1)} = z_{1:n}$ and propose $z_{1:n}^{(2)}, \ldots, z_{1:n}^{(N)} \sim \mu(\cdot)$.}
\Else{
Set $z_{1:n}^{(1)} = z_{1:n}$.

Sample (without replacement) a subset $b = \{ b_{1}, \ldots, b_{m} \} \subset \{1, \ldots, n\}$ uniformly.

\For{$i = 2, \ldots, N$}{
Set $z_{/b}^{(i)} = z_{/b}$, propose $z_{b}^{(i)} \sim \prod_{i = 1}^{m} \mu_{b_{i}}(\cdot)$, and set  $z_{1:n}^{(i)} = (z_{b}^{(i)}, z_{/b})$.}
}
Sample $k$ with probability proportional to $h_{\epsilon}(y \vert z_{1:n}^{(k)}, \theta')$. 

Accept $\theta', z_{1:n}^{(k)}$ as the new sample with probability 
\[
\min \left\{ 1,  \frac{q(\theta \vert \theta')}{q(\theta' \vert \theta)} \frac{\eta(\theta')}{\eta(\theta)} \frac{\sum_{i = 1}^{N} h_{\epsilon}(y \vert z_{1:n}^{(i)}, \theta')}{\sum_{i = 1}^{N} h_{\epsilon}(y \vert z_{1:n}^{(i)}, \theta)} \right\};
\]
otherwise reject and repeat $(\theta, z_{1:n})$ as the new value.

\end{algorithm}


\subsection{Exact inference based on sequential releases} \label{sec: Exact inference based on sequential releases}

Here we consider the scenario in Section \ref{sec: Fisher information with sequential release}, where the individuals' data are obtained in privacy preserving noise as in \eqref{eq: sequential release}. Here we have independent sequential observations $Y_{i}$, whose generative model involves a latent variable, $X_{i}$. As in Section \ref{sec: Exact inference based on the true posterior}, we will adopt the representation of the generative model in terms of random variables $Z_{i}$ whose distribution does not depend on $\theta$. Specifically, we assume that \eqref{eq: pseudorandom variables} holds for some distribution $\mu(\cdot)$ and functions $\varphi_{\theta}(\cdot)$ for all $\theta$. Again, as long as one can sample from $\mathcal{P}_{\theta}$, existence of such $Z$ is secured. This enables the joint distribution
\begin{equation}\label{eq: true posterior sequential}
\pi(\theta, z_{1:n} \vert y_{1:n}) \propto \eta(\theta) \prod_{t = 1}^{n} \mu(z_{t}) h_{\epsilon}(y_{t} \vert z_{t}, \theta),
\end{equation}
where $h_{\epsilon}(y_{t} \vert z_{t}, \theta) = g_{\epsilon}(y_{t} \vert s(\varphi_{\theta}(z_{t})))$. We aim to sample from the posterior distribution in \eqref{eq: true posterior sequential} using Algorithm \ref{alg: MHAAR for ABC-DP sequential}, which we adopt from a recently developed MHAAR algorithm in \citet{andrieu_et_al_2020}. The reason we choose this particular algorithm is that the variance of its acceptance ratio does not increase with $n$ as long as the proposal distribution for $\theta$ is properly scaled with the data size $n$; see \citet{yildirim_et_al_2018} for a related result. Note that this is in contrast to the PMMH algorithm of \citet{andrieu_and_roberts_2009} whose acceptance rate increases with $n$, leading to higher rejection rates, hence to slowing of the algorithm.

\begin{algorithm}
\caption{MHAAR for \eqref{eq: true posterior sequential} - one iteration}
\label{alg: MHAAR for ABC-DP sequential}
\KwIn{Current sample $(\theta, z_{1:n})$, subset size $m < n$, number of samples  for $z_{1:n}$: $N$, privately shared sequence: $y_{1:n}$.}
\KwOut{New sample}
Propose $\theta' \sim q(\cdot \vert \theta)$

\For{$t = 1,\ldots, n$}{
Set $z_{t}^{(1)} = z_{t}$ and propose $z_{t}^{(2)}, \ldots, z_{t}^{(N)} \sim \mu(\cdot)$.}

Calculate the acceptance probability 
\[
\alpha = \min \left\{ 1,  \frac{q(\theta \vert \theta')}{q(\theta' \vert \theta)} \frac{\eta(\theta')}{\eta(\theta)} \frac{\prod_{t = 1}^{n} \sum_{i = 1}^{N} h_{\epsilon}(y_{t} \vert z_{t}^{(i)}, \theta')}{\prod_{t = 1}^{n} \sum_{i = 1}^{N} h_{\epsilon}(y_{t} \vert z_{t}^{(i)}, \theta)} \right\}.
\]

Sample $v \sim \textup{Unif}(0, 1)$.\\
\If{$v < \alpha$}{
Return $(\theta', z_{1:n} = (z_{1}^{(k_{1})}, \ldots, z_{n}^{(k_{n})}))$, where each $k_{t} \in \{1, \ldots, N\} $ is sampled with probability proportional to $h_{\epsilon}(y_{t} \vert z_{t}^{(k_{t})}, \theta')$.
}
\Else{
Return $(\theta, z_{1:n} = (z_{1}^{(k_{1})}, \ldots, z_{n}^{(k_{n})}))$, where each $k_{t} \in \{1, \ldots, N\} $ is sampled with probability proportional to $h_{\epsilon}(y_{t} \vert z_{t}^{(k_{t})}, \theta)$.
}

\end{algorithm}

\section{Numerical examples} \label{sec: Numerical examples}
In this section we show some numerical examples which justify the proposed way of choosing the statistic of the sensitive data, as well as demonstrate the performance of the MCMC algorithms proposed for the different privacy settings that are described in Section \ref{sec: Statistic selection based on Fisher information}.

For Bayesian inference, a method for statistic selection can be reasonably justified if it yields the statistic that results in smallest MSE for the posterior expectation $\hat{\theta}(Y) = \mathbb{E}(\theta \vert Y)$, that is,
\begin{equation} \label{eq: MSE of posterior estimates}
\textup{MSE} = \mathbb{E}_{Y}[(\hat{\theta}(Y)-\theta^{\ast})^2], 
\end{equation}
In our experiments we will follow that way of justification for our statistic selection method based on the Fisher information. For a given $y$, $\hat{\theta}(y)$ will be obtained by one of the MCMC algorithms presented in Section \ref{sec: Bayesian inference using MCMC}, depending on the nature of the data generation model. MSE in \eqref{eq: MSE of posterior estimates} will approximated by $\textup{MSE} \approx \frac{1}{M} \sum_{i = 1}^{M} (\hat{\theta}(Y^{(i)})-\theta^*)^2$, where the $M$ independent samples for $Y^{(i)}$ are drawn conditional on the true value $\theta^{\ast}$. 

\subsection{Comparison of additive statistics with the Gaussian mechanism} \label{sec: Comparison of additive statistics with the Gaussian mechanism}
Our first example refers to the setting in Example \ref{ex: Variance parameter of the normal distribution}, where $\mathcal{P}_{\theta} = \mathcal{N}(0, \theta)$ with the natural limits $[-A, A]$ for the data and we computed $F(\theta)$ when the noisy statistic in \eqref{eq: noisy statistic} is constructed from $s(x) = \vert x\vert$ and $s(x) = x^{2}$, separately. In Example \ref{ex: Variance parameter of the normal distribution} we showed that $s(x) = \vert x\vert$ results in larger $F(\theta)$ than $s(x) = x^{2}$ when $\epsilon = 1$, while $s(x) = x^{2}$ becomes more informative when there is no privacy.

Here we compared the choices $s(x) = \vert x\vert$ and $s(x) = \vert x^{2} \vert$ in terms of MSE at various values of $\epsilon$. We took $A = 10$, $n = 100$, and $\theta^{\ast} = 2$. For MSE calculations, we took $M = 10^{3}$. To obtain the posterior expectations, we ran Algorithm \ref{alg: MH for ABC-DP}, with flat prior on $\theta$, to generate a total of $K = 10^{5}$ iterations and took the sample average after discarding the first quarter as burn-in.

The results are summarized in Figure \ref{fig: Alg1-NormFIM}. We observe that $s(x) = \vert x\vert$ outperforms $s(x) = x^2$ in terms of MSE unless $\epsilon$ is very large. Critically, we observe that when $F(\theta)$ is larger we have smaller MSE, which justifies the use of Fisher information for statistic selection.
\begin{figure}[ht]
\centerline{
\includegraphics[scale = 1.2]{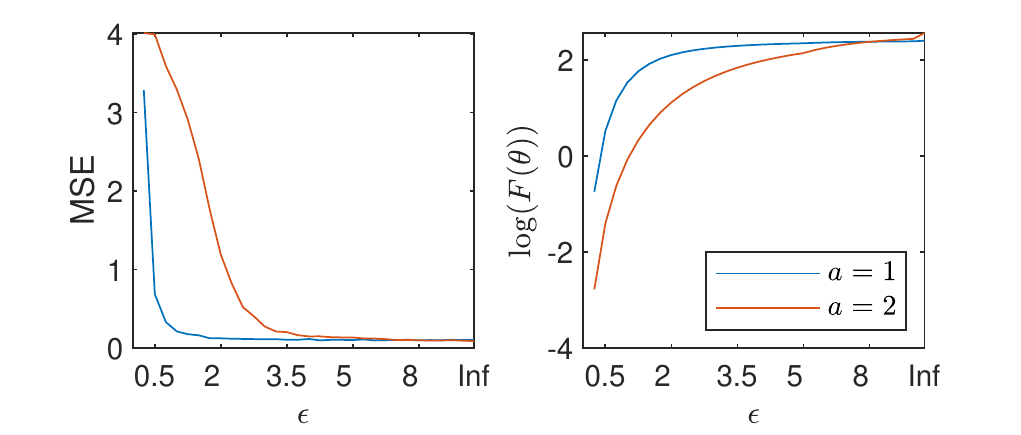} 
}
\caption{MSE for Algorithm \ref{alg: MH for ABC-DP} and (Logarithm of) $F(\theta)$ for different moments when there is Gaussian noise.}
\label{fig: Alg1-NormFIM}
\end{figure}

\subsection{Comparison of additive statistics with the Laplace mechanism} \label{sec: Comparison of additive statistics with the Laplace mechanism}
In this part, we repeat the previous experiment but with the following differences: We consider the Laplace mechanism, where the additive statistic is corrupted by Laplace noise as
\[
Y = \frac{1}{n}\sum_{i=1}^{n} \vert x_{i}\vert^{a} + V, \quad V \sim \text{Laplace}(0, A^{a}/(n\epsilon)), 
\]
For Bayesian inference, we used Algorithm \ref{alg: PMMH for ABC-DP}. Note that one could use Algorithm \ref{alg: MHAAR for ABC-DP} as well, which would yield the same qualitative results in terms of MSE. 

Figure \ref{fig: Laplace_Add-Norm_FIM} shows MSE, obtained with $M = 100$ noisy observations, and $F(\theta)$ for the choices $s(x) = \vert x\vert$  and $s(x) = x^2$. We observe that, like in the case where we use Gaussian mechanism, $s(x) = \vert x\vert$  provides more information than $s(x) = x^{2}$ under the Laplace noise. Moreover, MSE values and $F(\theta)$ are consistent also in this problem.
\begin{figure}[ht]
\centerline{
\includegraphics[scale = 1.2]{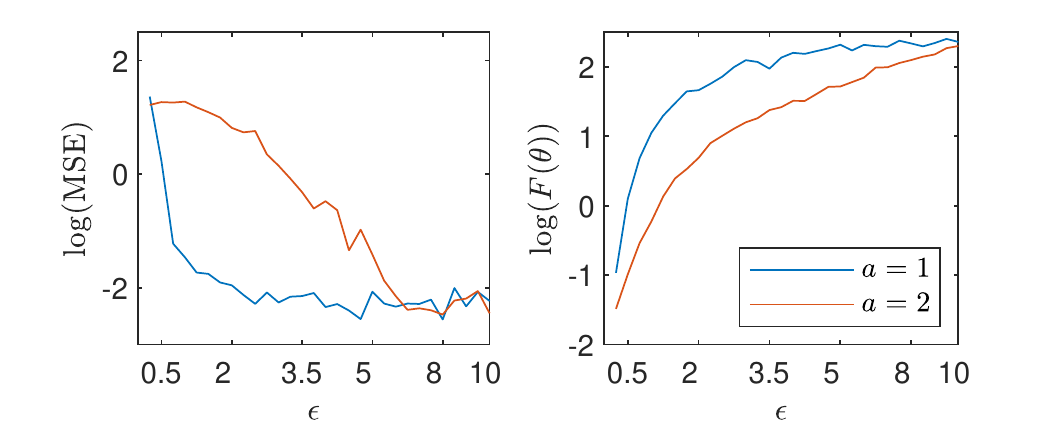}
}
\caption{MSE (left) and $F(\theta)$ (right) for $s(x) = \vert x\vert$ (blue) and $s(x) = x^{2}$ (red), under Laplace mechanism. MSE is calculated from the samples obtained from Algorithm \ref{alg: PMMH for ABC-DP}.}
\label{fig: Laplace_Add-Norm_FIM}
\end{figure}

\subsection{Comparison of Algorithms \ref{alg: PMMH for ABC-DP} and \ref{alg: MHAAR for ABC-DP} in terms of mixing}

 Recall that Algorithms \ref{alg: PMMH for ABC-DP} and \ref{alg: MHAAR for ABC-DP} are instances of PMMH and MHAAR, respectively, that both target the same posterior distribution in Section \ref{sec: MH for additive statistic and non-gaussian noise}. In this part we compare their performances in terms of integrated auto-correlation (IAC) time for $\theta$, which is the asymptotic variance of an average of samples generated by the MCMC algorithm relative to that of the average of i.i.d samples from the target distribution. (Hence, smaller IAC time is preferable.)

We continue with the setting in Section \ref{sec: Comparison of additive statistics with the Laplace mechanism}. We compared the IAC times of the algorithms with $s(x) = \vert x\vert$ and $\epsilon = 5$. 

For Algorithm \ref{alg: PMMH for ABC-DP}, the importance sampling distribution for $u$ was selected as $q_{\theta}(u) = f_{S_{n}}(u \vert \theta)$. For Algorithm \ref{alg: MHAAR for ABC-DP}, we chose the symmetric proposal distribution for $u$ as $q_{\theta, \theta'}(u) = f_{S_{n}}(u \vert (\theta + \theta')/2)$. In both algorithms, we used the same flat prior and the same random walk proposal for $\theta$.

Table \ref{tIAC of Algs} shows for both algorithms the IAC times vs sample size $N$ to estimate the acceptance ratios. As can be seen from Table \ref{tIAC of Algs}, Algorithm \ref{alg: MHAAR for ABC-DP} outperforms Algorithm \ref{alg: PMMH for ABC-DP} with lower IAC values for almost all of the $N$'s, while the performance gap closes as $N$ increases.

\begin{table}[ht]
\centering
\caption{IAC values of Algorithms \ref{alg: PMMH for ABC-DP} and \ref{alg: MHAAR for ABC-DP}}
\label{tIAC of Algs}
\begin{tabular}{ c  c  c }
\toprule
$N$ & Algorithm \ref{alg: PMMH for ABC-DP} & Algorithm \ref{alg: MHAAR for ABC-DP} \\ \midrule
2  & 44.03 & 17.99 \\
5 & 28.19 & 17.10  \\ 
10 & 21.11 & 16.13  \\ 
20 & 18.16 & 15.44  \\ 
50 & 15.32 & 13.78  \\ 
100 & 16.42 & 15.86  \\ 
\hline
\end{tabular}
\end{table}

\subsection{Inference based on a non-additive statistic} \label{sec: Inference based on non-additive statistic}
In this part we demonstrate the use of statistic selection method as well as the inference method when the compared statistics $S_{n}(X_{1:n})$ are non-additive. Specifically, we choose the maximum of $s(x_{i})$'s
\begin{align} \label{eq: max function of s}
S_{n}(X_{1:n}) &= \max \{ s(X_{i}); i = 1, \ldots, n \},
\end{align}
and the median of $s(x_{i})$'s
\begin{align}
S_{n}(X_{1:n}) &= \text{median} \{ s(X_{i}); i = 1, \ldots, n \} \label{eq: median function of s}
\end{align}
as two competitors for the statistic to be shared privately. As discussed in Section \ref{sec: Fisher information for based on the true marginal distribution}, adding noise to the maximum and median based on the global sensitivity is ineffective, because the global sensitivity of the those functions are determined by the range of $s(\cdot)$ irrespective of $n$. Instead, we consider generating the noisy statistic $Y$ by adjusting the amount of noise using the smooth sensitivity of the maximum and median functions. 

The smooth sensitivity formulas for the maximum and median can be found in \citet{Nissim_et_al_2007}, we give them here for the sake of completeness. Let $A_{s} = \max_{x \in \mathcal{X}} s(x)$ and assume $\min_{x \in \mathcal{X}} s(x) = 0$. (Otherwise $s(\cdot)$ can be shifted by a constant so that the minimum of their range is $0$.) Given the function $s(\cdot)$ and $x_{1}, \ldots, x_{n}$, let $s_{1}, \ldots, s_{n}$ be the sorted values of $s(x_{1}), \ldots, s(x_{n})$ so that $0 \leq s_{1} \leq s_{2} \leq \ldots \leq s_{n} \leq A_{s}$. For the maximum in \eqref{eq: max function of s}, the smooth sensitivity is given by
\begin{align*}
&\Delta^{smooth}_{\max, \beta}(x_{1:n}) = \max \{ e^{-k \beta} b_{k}; k = 0, \ldots, n\},
\end{align*}
with $b_{k} = \max\{ A_{s} - s_{n-k}, s_{n} - s_{n-k-1} \}$. For the median in in \eqref{eq: median function of s}, the smooth sensitivity is
\begin{align*}
\Delta^{smooth}_{\text{med}, \beta}(x_{1:n}) = \max \{ e^{-k \beta} b_{k};  k = 0, \ldots, n \}
\end{align*}
with $b_{k} = \max\{ s_{m+i} - s_{m+i-k-1}; i = 0, \ldots, k+1 \}.$

As in the previous examples, we have the same population distribution, $\mathcal{P}_{\theta} = \mathcal{N}(0, \theta)$, and the data generation process limits $X$'s to $[-A, A]$.

We ran Algorithm \ref{alg: MHAAR for ABC-DP with z} with each of the above choices for $S_{n}$ with $s(x) = \vert x\vert$. We took $\theta = 2$, $n=100$, and the differential privacy parameters are taken as $(\epsilon = 5, \delta = 1/n^{2})$. Table \ref{tMedian_max} shows the MSE values obtained with $M = 100$. We also report in Figure \ref{fig: MHAAR for ABC-DP with z FIM} the estimates of $F(\theta)$ for the median and maximum statistics, obtained with Algorithm \ref{alg: Monte Carlo estimation of FIM_non additive}, for various values of $\theta$.

\begin{table}[ht]
\centering
\caption{MSE for median and maximum statistics}
\label{tMedian_max}
\begin{tabular}{ l  c  c }
\toprule
$S_{n}(X_{1:n})$ & MSE  \\
\midrule
median & 0.391 \\
max & 22.64 \\ 
\hline
\end{tabular}
\end{table}

\begin{figure}[ht]
\centerline{
\includegraphics[scale = 1]{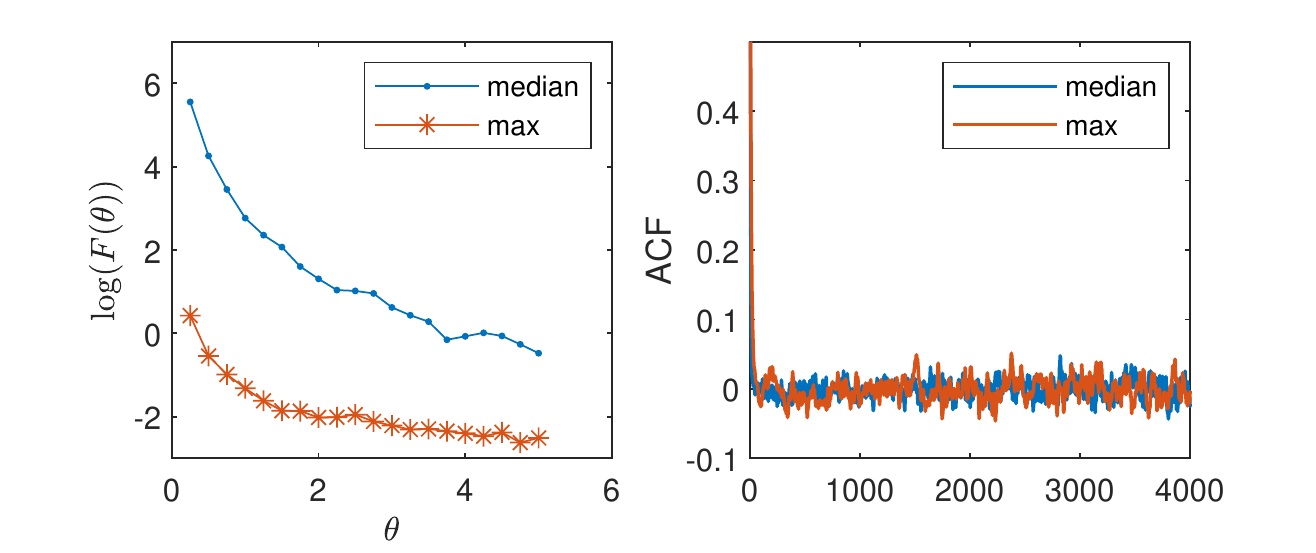}
}
\caption{$F(\theta)$ (left) for median (blue) and maximum (red) of $s(x) = \vert x\vert$, Auto-correlation function (ACF) for Algorithm \ref{alg: MHAAR for ABC-DP with z} for median (blue) and maximum (red). Privacy parameters are $(\epsilon, \delta) = (5, 1/n^{2})$.}
\label{fig: MHAAR for ABC-DP with z FIM}
\end{figure}

By observing $F(\theta)$ values and MSE values in Figure \ref{fig: MHAAR for ABC-DP with z FIM} and Table \ref{tMedian_max}, we can see that empirical results agree with the theoretical expectations. In Figure \ref{fig: MHAAR for ABC-DP with z FIM}, in terms of $F(\theta)$, median has better performance since it is definitely more informative as it can be interpreted from $F(\theta)$ values. Also, MSE values reveal that estimates obtained by using median statistics are closer and variate less from the desired parameter even in the non-additive and non-gaussian case.

Figure \ref{fig: MHAAR for ABC-DP with z FIM} shows the sample auto-correlation function, averaged over $5$ runs each with an independent noisy observation, for the median and maximum for $\vert x_{i} \vert$'s. We observe from the plots that Algorithm \ref{alg: MHAAR for ABC-DP with z} mixes well for both statistics, which suggests that the MSE calculations are reliable.

\subsection{Comparison of statistics in sequential release}
In this part, we utilize Algorithm \ref{alg: MHAAR for ABC-DP sequential} to compare statistics using sequential release. Laplace mechanism and normal posterior distribution with unknown variance is again the target in this case. However, as it was described in Section \ref{sec: Exact inference based on sequential releases}, algorithm aims to draw samples from individual noisy data points instead of summary statistics such as mean or median. 

Comparison of $s(x) = \vert x\vert$ and $s(x) = x^{2}$ are represented in Figure \ref{fig: MSE sequential Release}. We deduce from the figure that that $s(x) = \vert x \vert$ yields smaller MSE, as predicted by $F(\theta)$. 
\begin{figure}[ht]
\centerline{
\includegraphics[scale = 1.2]{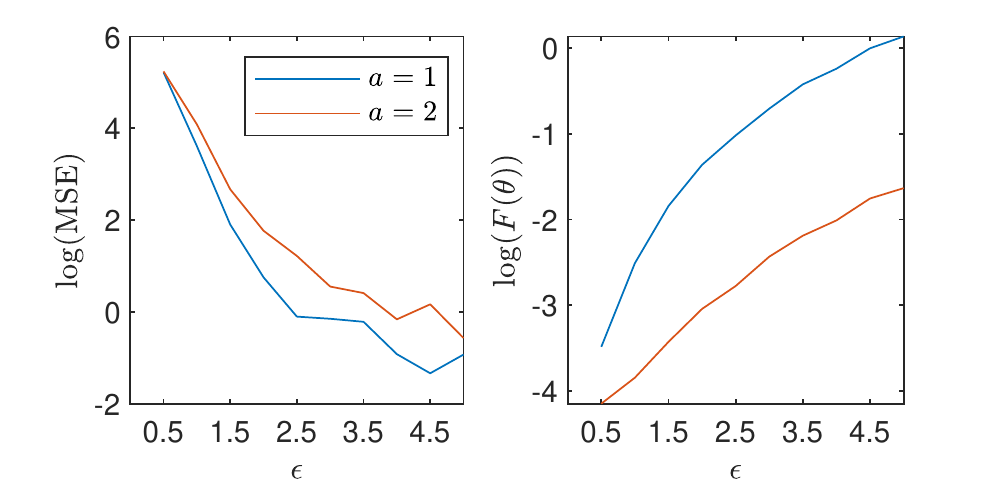}
}
\caption{MSE (left) and $F(\theta)$ (right) for $s(x) = \vert x\vert$ (blue) and $s(x) = x^{2}$ (red), under Laplace mechanism using sequential release. MSE is calculated from the samples obtained from Algorithm \ref{alg: MHAAR for ABC-DP sequential}.}
\label{fig: MSE sequential Release}
\end{figure}

\section{Conclusion} \label{sec: Conclusion}
In this paper, we propose a method for statistic selection for parameter estimation in a data privacy context. The method is based on the Fisher information. When one candidate statistics are not uniformly better than the other in terms of its Fisher information, the prior information for the parameter can be incorporated to make a final decision. To calculate the Fisher information, we propose several Monte Carlo algorithms for various data-sharing scenarios depending on the nature of the statistic and the privatization mechanism. We equip the statistic selection method with suitable MCMC algorithms for Bayesian parameter estimation given the shared (noisy) statistics of data. Our findings showed the usefulness of the statistic selection based on the Fisher information as well as the effectiveness of the proposed MCMC algorithms.

The proposed framework for selecting the statistic to be privately shared is not presented as a competitor of differentially private estimation methods. In principle, the method can be useful and incorporated into any likelihood-based parameter estimation algorithm by providing the most informative statistic among those considered. Bayesian estimation via MCMC is adopted and promoted in this paper not only due to offering incorporation of the prior distribution for $\theta$ but also for the breadth of the models for which it can be applied. However, statistic selection based on Fisher information can also be utilised for differentially private maximum likelihood estimation via EM as in \citet{Gong_2019}. Moreover, the work developed in this paper can be used in the schemes of \citet{Dwork_and_Smith_2010}, where several estimators, obtained from batches, are combined into one private estimator. 

The methodology presented in this paper is not specific to additive mechanisms in differential privacy. Moreover, it also extends to other definitions of privacy. Specifically, a privacy preserving mechanism can be constructed to satisfy a certain privacy level with respect to a privacy definition. The necessary condition for the presented methodology to be applied for that mechanism is the ability to write the conditional distribution of the generated output given the sensitive data.

One limitation of the work arises from the possibility of one Fisher information matrix not being greater than the other (in the sense of the difference being positive definite). In such a case, an alternative overall measure such as the trace of the Fisher information can be considered.

One possible extension of this work is adaptive clipping method in an online estimation setting where individuals' data are entered into the system sequentially and one-by-one. In such a case, each individual data can be received after clipping (so that the sensitivity is shrunk). The range of clipping can be determined in an adaptive way based on the data received so far. Adaptive clipping is already used for differentially private gradient-based algorithms \citep{Pichapati_et_al_2019, Andrew_et_al_2021}. It would be interesting to compare those methods to one that applies clipping to maximize informativeness of the clipped data.

\section*{Supplementary material}

The code to generate the numerical results in this paper can be found at \\
\url{https://github.com/barisalparslan1/Statistic\_Selection\_and\_MCMC}.
\section*{Acknowledgments}
The study was funded by the Scientific and Technological Research Council of Turkey (TÜBİTAK) ARDEB Grant No 120E534.

\begin{appendices}

\section{Proof of Proposition \ref{prop: MHAAR variant invariance}} \label{sec: Proof of Proposition}

\begin{proof}[Proof of Proposition \ref{prop: MHAAR variant invariance}]
We will prove the Proposition for the more general version where a subset of $z_{1:n}$ is updated.

Fix a subset $b \subseteq \{ 1, \ldots, n \}$. Let $z := z_{1:n}$ and $\mu(z) = \prod_{t = 1}^{n} \mu(z_{t})$ for a short-hand notation. Consider the joint distribution
\begin{align*}
&\pi_{b}(\theta, \theta', z^{(1:N)}, k)  :=  \eta(\theta) \mu(z^{(1)}) h(y \vert z^{(1)}, \theta) q(\theta' \vert \theta) \prod_{i = 2}^{N} R_{b}(z^{(i)} \vert z^{(1)}) \frac{h(y \vert z^{(k)}, \theta')}{\sum_{k'} h(y \vert z^{(k')}, \theta')} 
\end{align*}
where $R_{b}(\cdot \vert \cdot)$ is some conditional distribution whose selection will prove critical. 

Finally, let $B$ be the random variable corresponding to the subset $b$ whose probability distribution is denoted by $\xi(b) = \mathbb{P}(B = b)$. Consider the extended distribution 
\[
\pi(\theta, \theta', z^{(1:N)}, k) = \sum_{b \subseteq \{1, \ldots, n\}} \xi(b) \pi_{b}(\theta, \theta', z^{(1:N)}, k)
\]
The important point about $\pi(\theta, \theta', z^{(1:N)}, k)$ is that the marginal probability density of $\theta, z^{(1)}$ is the desired posterior distribution in \eqref{eq: true posterior} evaluated at $\theta, z^{(1)}$ and the rest of the variables are the auxiliary variables to enable a tractable MCMC algorithm. Therefore, one can sample from $\pi(\theta, \theta', z^{(1:N)}, k)$ and consider the components $\theta, z^{(1)}$, in particular the former, as samples from the true posterior distribution.

We show that when $B = b$ is sampled, Algorithm \ref{alg: MHAAR for ABC-DP with z} targets $\pi_{b}(\theta', z^{(2:N)}, k \vert \theta, z^{(1)})$. Its proposal mechanism of  corresponds to sampling $\theta', z^{(2:N)}, k$ from their conditional distribution $\pi_{b}(\theta', z^{(2:N)}, k \vert \theta, z^{(1)})$ and proposing the swapping
\[
\theta \leftrightarrow \theta', \quad z^{(1)} \leftrightarrow z^{(k)}.
\] 
The resulting acceptance ratio is 
\begin{align*}
& \hspace{-2cm} \frac{\pi_{b}(\theta', \theta, z^{(k)}, z^{(1:k-1)},  z^{(k+1:N)}, k) }{\pi_{b}(\theta, \theta', z^{(1)}, \ldots, z^{(N)}, k)} \\
&= \frac{q(\theta \vert \theta') \eta(\theta') \mu(z^{(k)}) h(y \vert z^{(k)}, \theta') }{q(\theta' \vert \theta) \eta(\theta) \mu(z^{(1)}) h(y \vert z^{(1)}, \theta) }  \frac{\prod_{i \neq k}^{N} R_{b}(z^{(i)} \vert z^{(k)}) \frac{h(y \vert z^{(1)}, \theta)}{\sum_{i = 1}^{N} h(y \vert z^{(i)}, \theta)}}{\prod_{i = 2}^{N} R_{b}(z^{(i)} \vert z^{(1)}) \frac{h(y \vert z^{(k)}, \theta')}{\sum_{i = 1}^{N} h(y \vert z^{(i)}, \theta')}} \\
&= \frac{q(\theta \vert \theta') \eta(\theta') \mu(z^{(k)}) h(y \vert z^{(k)}, \theta') }{q(\theta' \vert \theta) \eta(\theta) \mu(z^{(1)}) h(y \vert z^{(1)}, \theta) } \frac{\prod_{i \neq k}^{N} R_{b}(z^{(i)} \vert z^{(k)}) \frac{h(y \vert z^{(1)}, \theta)}{\sum_{i = 1}^{N} h(y \vert z^{(i)}, \theta)}}{\prod_{i = 2}^{N} R_{b}(z^{(i)} \vert z^{(1)}) \frac{h_{\theta'}(y \vert z^{(k)})}{\sum_{i = 1}^{N} h(y \vert z^{(i)}, \theta')}} \\
& = \frac{q(\theta \vert \theta') \eta(\theta') \mu(z^{(k)}) \prod_{i \neq k}^{N} R_{b}(z^{(i)} \vert z^{(k)}) }{q(\theta' \vert \theta) \eta(\theta) \mu(z^{(1)}) \prod_{i = 2}^{N} R_{b}(z^{(i)} \vert z^{(1)}) } \frac{\sum_{i = 1}^{N} h(y \vert z^{(i)}, \theta')}{\sum_{i = 1}^{N} h(y \vert z^{(i)}, \theta)}
\end{align*}
If the distribution $\mu(z^{(1)}) \prod_{i =2 }^{N} R_{b}(z^{(i)} \vert z^{(1)})$ is exchangeable with respect to $z^{(1:N)}$, then the acceptance ratio above simplifies to 
\[
\frac{q(\theta \vert \theta') \eta(\theta')  }{q(\theta' \vert \theta) \eta(\theta)} \frac{\sum_{k' = 1}^{N} h(y \vert z^{(k')}, \theta')}{\sum_{k' = 1}^{N} h(y \vert z^{(k')}, \theta)}.
\]
The proposal mechanism for the $z$ variable in Algorithm \ref{alg: MHAAR for ABC-DP with z}, which corresponds to $R_{b}$ here, satisfies the exchangeability property just mentioned. Hence, conditional on $B = b$, one iteration of Algorithm \ref{alg: MHAAR for ABC-DP with z} targets $\pi_{b}(\theta, \theta', z^{(1:N)}, k)$.

The proof is complete by observing that one iteration of Algorithm \ref{alg: MHAAR for ABC-DP with z} targets a $\pi_{b}(\theta, \theta', z^{(1:N)}, k)$ with probability $\xi(b)$, hence it targets $\pi(\theta, \theta', z^{(1:N)}, k)$.
\end{proof}

\end{appendices}

\bibliographystyle{apalike}
\bibliography{my_refs_arxiv}


\end{document}